\documentclass[acmsmall]{acmart}

\usepackage[T1]{fontenc}
\usepackage{graphicx}
\usepackage{subcaption}
\usepackage{pgfplots}
\usepackage{enumerate,paralist}
\usepackage{pbox}
\usepackage{stmaryrd}
\usepackage{comment}
\usepackage{parskip}
\usepackage{multirow}
\usepackage{url}

\usepackage[ruled, linesnumbered,vlined]{algorithm2e}
\usepackage{tikz}

\usepackage{calc}
\usepackage{fp}
\usepackage{wrapfig}

\usepackage{listings}
\usepackage{mathtools}
\usepackage{multibib}
\usepackage{flushend}

\newtheorem{remark}{Remark}

\newcommand{\expv}{\mathbb{E}}

\newcommand{\program}{\mathcal{P}}

\lstdefinelanguage{affprob}
{
morekeywords={angel,demon, choice, prob(0.6), prob(0.5), if, then, else, fi, 
while, do, od, 
true, false, and, or, skip, sample},
sensitive = false
}

\usepackage{tikz}
\usepackage{tikz,pgffor}
\usetikzlibrary{arrows}
\usetikzlibrary{shapes}
\usetikzlibrary{calc}
\usetikzlibrary{automata}
\usetikzlibrary{positioning}

\tikzstyle{ang}=[regular polygon, regular polygon sides = 3,draw,inner sep=0pt,minimum size=6mm, yshift = -0.75 mm]
\tikzstyle{dem}=[shape=diamond,draw,inner sep=0pt,minimum size=6mm]
\tikzstyle{ran}=[shape=circle,draw,inner sep=0pt,minimum size=5mm]
\tikzstyle{det}=[shape=rectangle,draw,inner sep=0pt,minimum size=5mm]
\tikzstyle{tran}=[draw,->,>=stealth, rounded corners]

\usepackage{todonotes}

\newcommand{\defineNote}[3][black!65!green]{\expandafter\def\csname #2\endcsname
##1{\stepcounter{fixcount}\fxwarning{\textcolor{#1}{\textbf{#3}: ##1}}}}

\newcommand{\theoremlike}[2]{\par\medskip\penalty-250%
{{\bfseries\noindent
#2 \ref{#1}.}}\it}

\newcommand{\thmhelperpre}[2]{\theoremlike{#1}{#2}}
\newcommand{\thmhelperpost}{\par\medskip}

\renewcommand{\vec}[1]{\mathbf{#1}}

\usepackage{graphicx}

\newcommand{\E}{\ensuremath{{\rm \mathbb E}}}

\renewcommand{\phi}{\varphi}

\newcommand{\eps}{\epsilon}

\newcommand{\pvars}{V}

\newcommand{\locs}{\mathit{L}}

\newcommand{\loc}{\ell}

\renewcommand{\E}{\mathbb{E}}

\newcommand{\Rset}{\mathbb{R}}
\newcommand{\Nset}{\mathbb{N}}
\newcommand{\Zset}{\mathbb{Z}}

\newcommand{\lin}{\mathit{in}}
\newcommand{\lout}{\mathit{out}}

\newcommand{\transitions}{\mapsto}
\newcommand{\probdist}{\mathit{Pr}}
\newcommand{\guards}{G}

\newcommand{\prob}{\mathit{Pr}}

\newcommand{\id}{\mathit{id}}
\newcommand{\probm}{\mathbb{P}}
\newcommand{\inv}{I}

\newcommand{\APP}{{\sc App}}
\newcommand{\PP}{{\sc PP}}

\newcommand{\support}{\mathit{supp}}

\newcommand{\vars}{\mathcal{V}}

\newcommand{\up}{u}
\newcommand{\pCFG}{\mathcal{C}}

\newcommand{\cfg}[2]{\vec{C}^{#1}_{#2}}
\newcommand{\lr}[2]{\loc^{#1}_{#2}}
\newcommand{\vr}[2]{\vec{x}^{#1}_{#2}}
\newcommand{\run}{\varrho}
\newcommand{\locinit}{\loc_{\mathit{init}}}
\newcommand{\vecinit}{\vec{x}_{\mathit{init}}}
\newcommand{\natfilt}{\mathcal{R}}
\newcommand{\lem}{\eta}
\newcommand{\preexp}[1]{\mathit{pre}_{#1}}

\newcommand{\confset}{Z}
\newcommand{\stime}{T}

\newcommand{\ttime}{\mathit{Term}}

\newcommand{\genPathSet}{A}

\newcommand{\updates}{\mathit{Up}}
\newcommand{\Fpath}{\mathit{Fpath}}
\newcommand{\Conf}{\mathit{Conf}}
\newcommand{\Run}{\mathit{Run}}
\newcommand{\genpath}{\pi}
\newcommand{\locsNB}{\locs_{\mathit{NB}}}
\newcommand{\locsPB}{\locs_{\mathit{PB}}}
\newcommand{\AppendixMaterial}{the supplementary material}
\newcommand{\genfilt}{\mathcal{F}}

\newcommand{\vecseq}[3]{\vec{#1}_{#2}[#3]}
\newcommand{\genRunSet}{\genPathSet}
\newcommand{\levelrank}[2]{\def\EmptyTest{#1}\ifdefempty{\EmptyTest}{\mathit{lev}_{#2}}{\mathit{lev}_{#2}(#1)}}
\newcommand{\minlev}{\mathit{min}\text{-}\mathit{lev}}
\newcommand{\noofdec}{\sharp}
\newcommand{\fixn}[1]{{#1}^*}
\newcommand{\indicator}[1]{{1}_{#1}}
\newcommand{\locterm}{\loc_{\mathit{term}}}
\newcommand{\programbody}{\program_{\mathit{body}}}
\newcommand{\slice}{\mathit{slice}}
\newcommand{\loops}{\mathit{loops}}
\newcommand{\OmegaRun}{\Omega_{\mathit{Run}}}

\newcommand{\vecinitset}{\Xi_{\mathit{init}}}
\newcommand{\linsystem}{\mathcal{L}}
\newcommand{\lp}{\mathcal{LP}}
\newcommand{\sol}{\mathit{sol}}
\newcommand{\noofdecrank}{\noofdec\mathit{lev}}
\newcommand{\nodecrank}[2]{\noofdec_{#1}\mathit{lev}_{#2}}

\newcites{add}{Additional References}
\begin{document}

\title{Lexicographic Ranking Supermartingales:\\ An Efficient Approach to 
Termination of Probabilistic Programs}

\author{Sheshansh Agrawal}
\affiliation{
\institution{IIT Bombay}
\city{Mumbai}
\country{India}
}
\email{sheshansh@cse.iitb.ac.in}

\author{Krishnendu Chatterjee}
\affiliation{
\institution{IST Austria} 
\city{Klosterneuburg}
\country{Austria}
}
\email{Krishnendu.Chatterjee@ist.ac.at}

\author{Petr Novotn\'{y}}
\affiliation{
\institution{IST Austria}
\city{Klosterneuburg}
\country{Austria}
}
\email{petr.novotny@ist.ac.at}

\begin{abstract}
Probabilistic programs extend classical imperative programs with 
real-valued random variables and random branching.
The most basic liveness property for such programs is the termination 
property.
The qualitative (aka almost-sure) termination problem given a probabilistic program
asks whether the program terminates with probability~1.
While ranking functions provide a sound and complete method for 
non-probabilistic
programs, the extension of them to probabilistic programs is achieved
via ranking supermartingales (RSMs). 
While deep theoretical results have been established about RSMs, 
their application to probabilistic programs with nondeterminism has been limited
only to academic examples. 
For non-probabilistic programs, lexicographic ranking functions provide a 
compositional
and practical approach for termination analysis of real-world programs. 
In this work we introduce lexicographic RSMs and show that they present a sound
method for almost-sure termination of probabilistic programs with nondeterminism.
We show that lexicographic RSMs provide a tool for compositional reasoning 
about almost sure termination,
and for probabilistic programs with linear arithmetic they can be synthesized 
efficiently (in polynomial time).
We also show that with additional restrictions even asymptotic bounds on expected
termination time can be obtained through lexicographic RSMs.
Finally, we present experimental results on abstractions of real-world programs
to demonstrate the effectiveness of our approach.
\end{abstract}

\maketitle


\terms{Program Verification, Termination}


\section{Introduction}\label{sec:introduction}

\noindent{\em Probabilistic programs with nondeterminism.} 
Randomness plays a fundamental role in many areas across science, and in 
computer science in particular.
In applications such as stochastic network protocols~\cite{BaierBook,prism},
randomized algorithms~\cite{RandBook,RandBook2}, 
security~\cite{BGGHS16:diff-privacy-coupling,BGHP16:diff-privacy-siglog} 
machine 
learning~\cite{LearningSurvey,G15},
the probabilistic behavior must be considered to faithfully model the underlying dynamic system.
The extension of classical imperative programs with \emph{random value generators}, 
that produce random values according to some desired probability distribution, 
naturally gives rise to probabilistic programs.
Along with probability, nondeterminism also plays a crucial role.
In particular in program analysis, for effective analysis of large programs,
all variables cannot be considered, and abstraction ignores some variables,
and the worst-case analysis is represented by adversarial nondeterminism.
Hence, probabilistic programs with nondeterminism have become an active and 
important research focus in program analysis.

\smallskip\noindent{\em Termination problem.} 
In static analysis of programs the most basic, as well most important, 
liveness property is the {\em termination} problem.
While for non-probabilistic programs the termination question asks whether
an input program {\em always} terminates, for probabilistic programs 
the termination questions must account for the probabilistic behaviors. 
The most basic and fundamental extensions of the termination problem 
for probabilistic programs are:

\begin{compactenum}
\item \emph{Almost-sure termination.} 
The \emph{almost-sure termination} problem asks whether the program terminates with probability~1.

\item \emph{Positive termination.} 
The \emph{positive termination} problem asks whether the expected termination time is finite.
A related quantitative generalization of the positive termination question is to obtain 
asymptotic bounds on the expected termination time.

\end{compactenum}
While the positive termination implies almost-sure termination, the converse is not true 
(e.g., see Example~\ref{ex:infinite-time}).

\smallskip\noindent{\em Ranking functions and ranking supermartingales (RSMs).}
The key technique that applies for liveness analysis of non-probabilistic programs is 
the notion of {\em ranking functions}, which provides a sound and complete 
method for termination of non-probabilistic programs~\cite{rwfloyd1967programs}.
There exist a wide variety of approaches for construction 
of ranking functions for non-probabilistic programs~\cite{DBLP:conf/cav/BradleyMS05,DBLP:conf/tacas/ColonS01,DBLP:conf/vmcai/PodelskiR04,DBLP:conf/pods/SohnG91}.
The generalization of ranking functions to probabilistic programs is achieved through the
{\em ranking supermartingales (RSMs)}~\cite{SriramCAV,HolgerPOPL,CF17}.
The ranking supermartingales provide a powerful and automated approach for termination 
analysis of probabilistic programs, and algorithmic approaches for special cases such as 
linear and polynomial RSMs have also been considered~\cite{SriramCAV,CFNH16:prob-termination,CFG16,CNZ17}.

\smallskip\noindent{\em Practical limitations of existing approaches.} 
While an impressive set of theoretical results related to RSMs has been
established~\cite{SriramCAV,HolgerPOPL,CF17,CFNH16:prob-termination,CFG16,CNZ17}, 
for probabilistic programs with nondeterminism the current approaches  
are only applicable to academic examples of variants of random walks. 
The key reason can be understood as follows: even for non-probabilistic programs while 
ranking functions are sound and complete, they do not necessarily provide a practical 
approach. This is because to prove termination, a witness in the form of a 
ranking function has to be computed: to this automatically, ranking functions 
of a restricted shape (such as linear ranking functions) have to be considered, 
and 1-dimensional ranking functions of a restricted type can only prove 
termination of a limited class of programs. 
In contrast, as a practical and scalable approach for non-probabilistic programs the 
notion of lexicographic ranking functions has been widely 
studied~\cite{CSZ13,ADFG10:lexicographic,GMR15:rank-extremal,BCIKP16:T2}. 
Algorithmic approaches for linear lexicographic ranking functions allow the 
termination analysis 
to be applicable to real-world non-probabilistic programs (after abstraction).
However both the theoretical foundations as well as practical approaches related to 
such lexicographic ranking functions are completely missing for probabilistic programs,
which we address in this work.

\smallskip\noindent{\em Our contributions.} 
In this work our main contributions range from theoretical foundations of lexicographic
RSMs, to algorithmic approaches for them, to experimental results
showing their applicability to programs.
We describe our main contributions below:
\begin{enumerate}

\item {\em Theoretical foundations.} 
First, we introduce the notion of lexicographic RSMs, and show that
such supermartingales ensure almost-sure termination (Theorem~\ref{thm:lexrsm-main} in 
Section~\ref{sec:lexicographic}). 
Our first result is a purely mathematical result that introduces a new concept, and 
proves almost-sure termination, that is independent of any probabilistic program.
Based on the mathematical result we show that for probabilistic programs with 
nondeterminism
the existence of a lexicographic RSM with respect to an invariant ensures
almost-sure termination (Theorem~\ref{thm:lexrsm-programs} in 
Section~\ref{sec:lex-programs}).

\item {\em Compositionality.} 
Second we study the compositional properties of lexicographic RSMs.
A key limitation of the previous approaches related to compositional 
RSMs~\cite{HolgerPOPL} is that it imposes a technical \emph{uniform 
integrability} conditions, which is hard to reason about automatically. We show 
(in Section~\ref{sec:compositional}) how 
lexicographic RSMs 
present an easy-to-automatize compositional approach for almost-sure 
termination of probabilistic programs.

\item {\em Algorithm.} We then consider algorithms for synthesis of lexicographic RSMs, 
and for efficient algorithms we consider nondeterministic probabilistic 
programs that are {\em affine} (i.e., the arithmetic operations are linear).
We present a polynomial-time algorithm for synthesis of lexicographic RSMs 
for affine programs (Theorem~\ref{thm:algo}).

\item {\em Asymptotic bounds.} 
In general, the existence of lexicographic RSMs does not imply positive
termination. 
In other words, we present an example (Example~\ref{ex:infinite-time}) where a 
lexicographic RSM exists ensuring
almost-sure termination, yet the expected termination time is infinite.
We then present a natural restriction under which the lexicographic 
RSMs not only imply positive termination, but even asymptotic bounds on the expected termination 
time can be derived from them (Theorem~\ref{thm:runtime-bound} and 
Corollary~\ref{col:runtime-progs}).

\item {\em Experimental results.} 
We present experimental results of our approach on realistic programs to show the applicability 
of our approach. 
To demonstrate the effectiveness of our approach we consider the benchmarks of non-probabilistic
programs from~\cite{ADFG10:lexicographic} which are obtained as abstraction of 
real-world programs, 
where 
lexicographic ranking functions were applied for termination analysis. 
We extend these benchmarks with probabilistic statements and apply lexicographic RSMs to these
programs. 
Our experimental results show that our approach can handle these programs very efficiently.

\end{enumerate}

\vspace{-1em}
\section{Preliminaries}\label{sec:prelim}

We use a standard notation in the field of probabilistic program 
analysis~\cite{CNZ17}.

\vspace{-1em}
\subsection{Basic Notions}
\vspace{-0.5em}
For a set $A$ we denote by $|A|$ the cardinality of $A$. We denote by $\Nset$,
$\Nset_0$, $\Zset$, and $\Rset$ the sets of all positive integers, non-negative
integers, integers, and real numbers, respectively. We assume basic knowledge 
of matrix calculus.
We use boldface notation for
vectors, e.g. $\vec{x}$, $\vec{y}$, etc., and we denote an $i$-th component of a
vector $\vec{x}$ by $\vec{x}[i]$. 
We identify 1-dimensional vectors with numbers. For an 
$n$-dimensional vector 
$\vec{x}$, index $1 \leq i\leq n$, and number $a$ we denote by $\vec{x}(i\leftarrow a)$ 
a 
vector $\vec{y}$ such that $\vec{y}[i]=a$ and $\vec{y}[j]=\vec{x}[j]$ for all 
$1\leq j \leq n$, $j\neq i$.
For comparison of vectors (e.g. as in $\vec{x}\leq \vec{y}$), we consider componentwise 
comparison. 
For comparing functions $f,g$ with the same domains, we write $f\leq g$ 
if $f(x)\leq g(x)$ for all $x$ in the domain.

\smallskip\noindent{\em Variables.}
Throughout the paper we fix a countable set of variables $\vars$. 
We consider some arbitrary but fixed linear order on the set of all variables, 
hence we write 
$\vars=\{x_1,x_2,x_3,\dots\}$. 

\vspace{-1em}
\subsection{Syntax of Probabilistic Programs}\label{subsec:syntax}
\vspace{-0.5em}

In this subsection we define the form of probabilistic programs that we 
consider in our analysis. We consider two classes of probabilistic programs: 
general probabilistic programs (\PP{}s ) with arbitrary (measurable) 
expressions and 
their subclass, affine probabilistic programs (\APP s) where all expressions 
are restricted to be affine (see below for a precise definition). The reason 
for this dual view is that our work also has two main points of focus: a 
theoretical one, where we introduce new proof rules that can be used to prove 
properties of general probabilistic programs; and an algorithmic one, where we 
aim to prove properties of probabilistic programs automatically, using the 
aforementioned proof rules. As already testified in the non-probabilistic 
world, programs that contain only affine expressions allow for more efficient 
automation of the analysis and at the same time, due to the presence of 
non-determinism they can be used to form sound abstractions of programs with 
non-linear arithmetic. Hence, we consider general programs when providing our 
theoretical results and \APP s when presenting the automation of our techniques.

\smallskip\noindent{\em Expressions.}
An \emph{expression} over the set of variables $\{x_1,\dots,x_n\}$ 
is an expression in the standard programming-language sense, i.e. a formula 
built in finite number of steps from constants, variables $x_1,\dots,x_n$, and 
numerical operators from some fixed finite set. Each expression $E$ over 
$\{x_1,\dots,x_n\}$ 
determines a function which for each $m$-dimensional vector $\vec{x}$, where 
$m\geq n$,  
returns a number resulting from substituting each $x_i$ in $E$ by $\vec{x}[i]$. 
Slightly abusing our notation, we denote this function also by $E$ and the 
value of this function on argument $\vec{x}$ by $E(\vec{x})$. We do not a priori fix a concrete set of operators that can be used to form expressions. However, in order to ensure that semantics of probabilistic programs with real-valued variables is defined correctly, we impose the following two conditions on the set of expressions used in each program:
(1)
For each expression $E$ over variables $\{x_1,\dots,x_n\}$ and each $n$-dimensional vector $\vec{x}$ the value $E(\vec{x})$ is well defined.\footnote{Our results can be easily extended to programs where encountering an expression of undefined value, such as division by zero, triggers an exception which terminates the program, but we abstract away from such details for the sake of clarity.}
(2)
The function defined by each expression $E$ is Borel-measurable (for definition of Borel-measura\-bility, see, e.g.~\cite{Billingsley:book}).

From measure theory it is known that these conditions hold in particular for programs where expressions are build using the standard arithmetic operators of addition, subtraction, multiplication, and division (provided that expressions evaluating to zero are not allowed as divisors).


\smallskip\noindent{\em Affine Expressions.}
An \emph{affine expression} over the set of variables $\{x_1,\dots,x_n\}$ is an 
expression of the form $d+\sum_{i=1}^{n}a_i
x_i$, where  $d,a_1,\dots,a_n$ are real-valued
constants.  A function of the form $E(\vec{x})$ for some affine expression $E$ 
is called affine. As noted above, each affine function is Borel-measurable.

\smallskip\noindent{\em Predicates.}
A \emph{predicate} is a logical formula obtained by a finite number of 
applications of conjunction, disjunction and negation operations on 
\emph{atomic predicates} of the form $E\leq E'$, where $E$, $E'$ are 
expressions. We denote by $\vec{x}\models E$ the fact that $E$ is satisfied by 
substituting values from of $\vec{x}$ for the corresponding variables in $E$.

\smallskip\noindent{\em Linear constraints, assertions, predicates.}
In the case of predicates involving only linear expression we use the following 
standard nomenclature:
\begin{compactitem}
	\item {\em Linear Constraint.} A \emph{linear constraint} is a formula of the
	form $\psi$ or $\neg\psi$, where  $\psi$ is a non-strict inequality
	between affine expressions.
	\item {\em Linear Assertion.} A \emph{linear assertion} is a finite conjunction
	of linear constraints.
	\item {\em Propositionally Linear Predicate.}
	A  \emph{propositionally linear predicate} (PLP) is a finite disjunction of
	linear assertions.
\end{compactitem}



\noindent{\em The Syntax of Probabilistic Programs {(\PP s)}.}
We consider the standard syntax for probabilistic programs,
which encompasses basic programming mechanisms such as assignment statement 
(indicated by `:='), while-loop, if-branch. Expressions appear on right-hand sides of assignments, and predicates act as loop guards and conditions in if-then-else statements. We also consider basic probabilistic mechanisms 
such as probabilistic branch (indicated by `prob') and random sampling (e.g. 
$x:=\textbf{sample(}\mathrm{Uniform}[-2,1]\textbf{)}$ assigns to $x$ a random 
number 
uniformly sampled from interval $[-2,1]$). We also allow constructs for 
(demonic) non-determinism, in particular 
non-deterministic branching indicated by `\textbf{if }$\star$ \textbf{then...}' construct and non-deterministic assignment.  
Variables (or identifiers) of a probabilistic program are of \emph{real} type, i.e., 
values of the variables are real numbers. 
We also assume that assume that each \PP{} $\program$ is preceded by a preamble 
specifying possible initial values of program variables: the preamble consists 
of a single predicate characterizing possible initial valuations.

\smallskip\noindent{\em Affine Probabilistic Programs (\APP s).}
A probabilistic program is \emph{affine} if all the expressions that occur in 
the program (i.e. in loop guards, conditionals, right-hand sides of 
assignments) are affine and if the set of possible initial valuations is a 
polyhedron. We refer to the class of affine probabilistic programs 
as \APP s.

Due to space restrictions, details of syntax (such as grammar) are relegated to 
\AppendixMaterial.
For an example see Figure~\ref{fig:invariant-running}.

\vspace{-1em}
\subsection{Semantics of Probabilistic Programs}\label{subsec:semantics}
\vspace{-0.5em}

We now formally define the semantics of \PP's.
In order to do this, we first recall some fundamental concepts from probability
theory.

\smallskip\noindent{\em Basics of Probability Theory.}
 A probability space is a triple
$(\Omega,\mathcal{F},\probm)$, where $\Omega$ is a non-empty set (so called
\emph{sample space}), $\mathcal{F}$ is a \emph{sigma-algebra} of measurable 
sets over $\Omega$,
i.e. a collection of subsets of $\Omega$ that contains the empty set
$\emptyset$, and that is closed under complementation and countable unions, and
$\probm$ is a \emph{probability measure} on $\mathcal{F}$, i.e., a function
$\probm\colon \mathcal{F}\rightarrow[0,1]$ such that
\begin{compactitem}
\item $\probm(\emptyset)=0$,
\item for all $A\in \mathcal{F}$ it holds $\probm(\Omega\smallsetminus
A)=1-\probm(A)$, and
\item for all pairwise disjoint countable set sequences $A_1,A_2,\dots \in
\mathcal{F}$ (i.e., $A_i \cap A_j = \emptyset$ for all $i\neq j$)
we have $\sum_{i=1}^{\infty}\probm(A_i)=\probm(\bigcup_{i=1}^{\infty} A_i)$.
\end{compactitem}

Following the usual probabilistic terminology, we say that \emph{almost all} $\omega$ belonging to some set $O\subseteq \Omega$ satisfy some property $\Psi$ if it holds that $\probm{(\{\omega\in O\mid \omega \text{ does not satisfy }\Psi\})}=0$.

\noindent{\em Random variables and filtrations.}
A \emph{random variable} in a probability space $(\Omega,\mathcal{F},\probm)$ is
an $\mathcal{F}$-measurable function $R\colon \Omega \rightarrow \Rset \cup
\{\infty\}$, i.e.,
a function such that for every $a\in \Rset \cup \{ \infty\}$ the set
$\{\omega\in \Omega\mid R(\omega)\leq a\}$ belongs to $\mathcal{F}$. If 
$R(\omega)\in \Rset$ for all $\omega\in \Omega$, we say that $R$ is 
\emph{real-valued.}
We denote by $\expv[R]$ the \emph{expected value} of a random variable $X$~(see \cite[Chapter 5]{Billingsley:book}
for a formal definition). 
A \emph{random vector} in $(\Omega,\mathcal{F},\probm)$ is a vector whose every component is a random 
variable in this probability space. A \emph{stochastic process} in a 
probability space $(\Omega,\mathcal{F},\probm)$ is an infinite sequence of 
random vectors in this space.
We will also use random variables of the form $R\colon\Omega \rightarrow S$ for some finite 
set $S$, which is easily translated to the variables above.
A \emph{filtration} of a sigma-algebra $\mathcal{F}$ is a
sequence $\{\mathcal{F}_i \}_{i=0}^{\infty}$ of $\sigma$-algebras 
such that $\mathcal{F}_0 \subseteq \mathcal{F}_1 \subseteq \cdots \subseteq
\mathcal{F}_n \subseteq \cdots \subseteq \mathcal{F}$.

\emph{Distributions.} We assume the standard definition of a probability 
distribution specified by a cumulative distribution 
function~\cite{Billingsley:book}. We denote by $\mathcal{D}$ be a set of 
probability distributions on 
real numbers, both discrete and continuous.

\smallskip\noindent{\em Probabilistic Control Flow Graphs.}
We consider standard operational semantics of \PP{}s defined via an 
uncountable state-space
Markov decision process (MDP) (uncountable due to real-valued variables).
That is, we associate to each program 
a certain stochastic process.
To define this process, we first define so called 
\emph{probabilistic control flow graphs}~\cite{CFG16}.

\smallskip
\begin{definition}
\label{def:stochgame}
A \emph{probabilistic control flow graph (pCFG)} is a tuple
$\pCFG=(\locs,\pvars,\locinit,\vecinitset,\transitions,\updates,\probdist,\guards)$,
where
\begin{compactitem}
\item $\locs$ is a finite set of \emph{locations} partitioned into four 
pairwise
disjoint subsets  $\locsNB$, $\locsPB$, $\locs_D$, and $\locs_A$ of 
non-deterministic branching, 
probabilistic branching, deterministic, and assignment locations;
\item $\pvars=\{x_1,\dots,x_{|\pvars|}\}$ is a finite set of \emph{program 
variables} (note that $\pvars \subseteq \vars$) ;
\item $\locinit$ is an initial location and $\vecinitset$ is a set of initial 
\emph{assignment vectors};
\item $\transitions\subseteq\locs\times\locs$ is a transition relation;
\item $\updates$ is a function assigning to each transition outgoing from an 
assignment location a tuple $(i,\up)$, where $1\leq i \leq |\pvars|$ is a 
\emph{target 
	variable index} and $\up$ 
is an 
\emph{update element}, which can 
be one of the following mathematical objects: 
(a)~a Borel-measurable function $u\colon \Rset^{|\pvars|}\rightarrow \Rset$;
(b)~a distribution $d\in \mathcal{D}$; or
(c)~a set $R\subseteq \Rset$ (representing a non-deterministic update).
\item $\probdist=\{\prob_{\ell}\}_{\ell \in \locsPB}$ is a collection of
probability distributions, where each $\prob_{\ell}$ is a discrete probability
distribution on the set of all transitions outgoing from~$\ell$;
\item $\guards$ is a function assigning a propositionally linear predicate
(a \emph{guard}) over $\pvars$ to each transition outgoing from a deterministic 
location.
\end{compactitem}

We assume that each location has at least one outgoing transition.
Also, for every deterministic location $\ell$ we assume the following: if
$\tau_1,\dots,\tau_k$ are all transitions outgoing from $\ell$, then $G(\tau_1)
\vee \dots \vee G(\tau_k) \equiv \mathit{true}$ and $G(\tau_i) \wedge G(\tau_j)
\equiv \mathit{false}$ for each $1\leq i < j \leq k$. For each 
distribution $d$ appearing in the 
pCFG we assume the following features are known: expected value $\expv[d]$ of 
$d$ and a set $SP_d$ containing the \emph{support} of $d$. Tthe support is the  
smallest 
closed set of real numbers whose complement has probability zero 
under $d$\footnote{In particular, a support of a \emph{discrete} probability 
	distribution $d$ is simply the at most countable set of all points on a 
	real 
	line that have positive probability under $d$.}) Finally, we assume that 
	each 
assignment location has at most (and thus exactly) one outgoing transition. The 
translation from probabilistic programs to the corresponding pCFG is 
standard~\cite{CFNH16:prob-termination}, and the details are presented 
in~\AppendixMaterial.
\end{definition}

\smallskip\noindent{\em Configurations.}
A \emph{configuration} of a pCFG $\pCFG$ is a tuple $(\ell,\vec{x})$,
where $\ell$ is a location of $\pCFG$ and $\vec{x}$ is an 
$|\pvars|$-dimensional vector.
We say that a transition $\tau$ is \emph{enabled} in a configuration
$(\ell,\vec{x})$ if $\ell$ is the source location of $\tau$ and in addition,
${\vec{x}}\models G(\tau)$ provided that $\ell$ is deterministic. 

\smallskip\noindent{\em Executions and reachable configurations.}
We say that a configuration $(\loc',\vec{x}')$ is a \emph{successor} of a 
configuration $(\loc,\vec{x})$ if 
there is a transition  $\tau=(\loc,\loc')$ enabled in
$(\loc,\vec{x})$ 
and $\vec{x}'$ satisfies the following:
\begin{compactitem}
\item if $\loc$ is not an assignment location, then $\vec{x}'=\vec{x}$;
\item if $\loc$ is an assignment location with $\updates(\tau)=(j,\up)$, then 
$\vec{x}_{i+1}=\vec{x}_i(j\leftarrow a)$ where $a$ satisfies one of the 
following 
depending on the type of $\up$:
\begin{compactitem}
\item if $\up$ is a Borel-measurable function, then $a=\up(\vec{x})$;
\item if $\up$ is an integrable\footnote{A distribution on some numerical 
domain 
is integrable if its expected value exists and is finite. In particular, each 
Dirac distribution is integrable.} distribution $d$, then 
$a\in \support(d)$;
\item if $\up$ is a set, then $a$ is some element of $\up$.
\end{compactitem}
\end{compactitem}
A \emph{finite path} (or
\emph{execution fragment}) of length $k$ in $\pCFG$ is a finite sequence of
configurations $(\ell_0,\vec{x}_0)\cdots(\ell_k,\vec{x}_k)$ such that 
$\loc_0=\locinit$, $\vec{x}_0\in\vecinitset$, and 
for each
$0 \leq i < k$ the configuration $(\loc_{i+1},\vec{x}_{i+1})$ is a successor of 
$(\loc_i,\vec{x}_i)$.
A \emph{run} (or \emph{execution}) in
$\pCFG$ is an infinite sequence of configurations whose every finite
prefix is a finite path.
A configuration $(\loc,\vec{x})$ is {\em reachable} from the initial 
configuration
$(\locinit,\vecinit)$ (where, $\vecinit\in \vecinitset$)
if there is a finite path starting in $(\locinit,\vecinit)$ that ends in
$(\loc,\vec{x})$. We denote by $\Conf_\pCFG,\Fpath_\pCFG$ and $\Run_\pCFG$ the 
sets of all configurations, finite paths and runs in $\pCFG$, respectively, 
dropping the index $\pCFG$ when known from the context.

\smallskip\noindent{\em Non-determinism and Schedulers.}
The probabilistic
behaviour of $\pCFG$ can be captured by constructing a suitable
probability measure over the set of all its runs. Before this can be
done, non-determinism in $\pCFG$ needs to be resolved. This is achieved using 
the 
standard notion of a \emph{scheduler}. Note that there are two sources of 
non-determinism in our programs: one in branching and 
one in assignments. We call a location $\loc$ non-deterministic if $\loc$ is a 
non-deterministic branching location or if $\loc$ is an assignment location 
with the only transition $\tau$ outgoing from 
$\loc$ having a non-deterministic assignment. A configuration 
$(\loc,\vec{x})$ is non-deterministic if $\loc$ is non-deterministic.

\smallskip
\begin{definition}[Schedulers]
\label{def:schedulers}
A scheduler in a pCFG $\pCFG$ is a function $\sigma$
 assigning 
to every finite path that ends in a non-deterministic configuration 
$(\loc,\vec{x})$ a probability distribution on successor configurations of 
$(\loc,\vec{x})$.

\end{definition}

\smallskip\noindent{\em Measurable schedulers.}
Note that schedulers can be viewed as partial functions from the set $\Fpath$ 
to the set of probability distributions over the set $\Conf$.
Since we deal with programs operating over real-valued variables, both $\Fpath$ 
and $\Conf$ can be uncountable sets. Hence, we impose an 
additional \emph{measurability} condition on schedulers, so as to ensure that 
the semantics of probabilistic non-deterministic programs is defined in a 
mathematically 
sound way. First we need to clarify what are measurable sets of configurations 
and histories. We define a sigma-algebra $\mathcal{F}_{\Conf}$ of measurable 
sets of configurations to be the sigma-algebra 
over $\Conf$ 
generated\footnote{In general, it is known that for each set $\Omega$ and each 
collection of its subsets $F\subseteq 2^{\Omega}$ there exists at least one 
sigma-algebra $\mathcal{F}$ s.t. $F\subseteq \mathcal{F}$ and the intersection 
of all such sigma-algebras is again a sigma algebra --  so called sigma-algebra 
generated by $F$~\cite{Billingsley:book}.} by 
all sets of the form $\{\loc\}\times B$, where $\loc$ is a location of $\pCFG$ 
and  
$B$ is a Borel-measurable subset 
of $\Rset^{|\pvars|}$. Next, the set of finite paths $\Fpath$ can be viewed as 
a subset of $\Conf\cup \Conf\times\Conf \cup \Conf\times\Conf\times\Conf \cup 
\dots$. Hence, we define the sigma-algebra $\mathcal{H}$ of measurable sets of 
finite paths 
to be the sigma algebra generated by all sets of the form $\confset_1\times 
\confset_2 \times 
\cdots \times \confset_k \subseteq \Fpath$ such that $k\in \Nset$ and 
$\confset_i \in 
\mathcal{F}_{\Conf}$ for 
all $1\leq i \leq k$. Now we can define the measurability of schedulers. Recall 
that for each finite path $\genpath$ ending in a non-deterministic 
configuration 
we have that $\sigma(\genpath)$ is a probability distribution on $\Conf$. For 
each measurable set of configurations $\confset$ we denote by 
$\sigma(\genpath)(\confset)$ the probability that the random draw from 
distribution $\sigma(\genpath)$ selects an element of $\confset$. We say that a 
scheduler $\sigma$ is measurable if for each $\confset \in \mathcal{F}_{\Conf}$ 
and each $p\in[0,1]$ the set $\{\pi \in \Fpath \mid 
\sigma(\genpath)(\confset)\leq p\}$ belongs to $\mathcal{H}$, i.e., it is a 
measurable set of paths.

While the definition of a measurable scheduler might seem somewhat technical, 
it is natural from measure-theoretic point of view and analogous definitions 
inevitably emerge in works dealing with systems that exhibit both probabilistic 
and non-deterministic behaviour over a continuous state 
space~\cite{NSK:CTMDP-delayed,NK:CTMDP-bisimulation}. In particular, if all the 
variables in a program range over a discrete set (such as the integers), 
then each scheduler in the associated pCFG is measurable.


\smallskip\noindent{\em Stochastic process.}
A pCFG $\pCFG$ together with a scheduler $\sigma$ and initial valuation 
$\vecinit\in\vecinitset$ define a stochastic 
process which produces a random run 
$(\loc_0,\vec{x}_0)(\loc_1,\vec{x}_1)(\loc_2,\vec{x}_2)\cdots$. The evolution 
of 
this process can be informally described as follows: we start in the initial 
configuration, i.e. $(\loc_{0},\vec{x}_0)=(\locinit,\vecinit)$. 
Now assume that $i$ 
steps have elapsed, i.e. a finite path 
$\genpath_i=(\loc_0,\vec{x}_0)(\loc_1,\vec{x}_1)\cdots(\loc_i,\vec{x}_i)$ has 
already 
been produced. Then a successor configuration $(\loc_{i+1},\vec{x}_{i+1})$ is 
chosen as follows:
\begin{compactitem}
\item
If $\loc_i$ is a non-deterministic location, then 
$(\loc_{i+1},\vec{x}_{i+1})$ is sampled according to scheduler $\sigma$, i.e. 
from 
the distribution $\sigma(\genpath_i)$.
\item
If $\loc_i$ is an assignment location (but not a non-deterministic one) with, 
there 
is exactly one transition $\tau=(\loc_i,\loc')$ outgoing from it and we put 
$\loc_{i+1}=\loc'$. Denoting $\updates(\tau)=(j,\up)$, the vector 
$\vec{x}_{i+1}$ is 
then defined as 
$\vec{x}_{i+1}=\vec{x}_{i}(j\leftarrow a)$ where $a$ is chosen depending on 
$\up$:
\begin{compactitem}
\item If $u$ is a function $u\colon
\Rset^{|\pvars|}\rightarrow \Rset$, then $a=f(\vec{x}_{i})$.
\item If $\up$ is a distribution $d$, then $a$ 
is sampled from $d$.
\end{compactitem}
\item In all other cases we have $\vec{x}_{i+1}=\vec{x}_i$, and $\loc_{i+1}$ is 
determined as follows:
\begin{itemize}
	\item If $\loc_i$ is a probabilistic branching location, then a transition 
	$(\loc_i,\loc')$ is sampled from $\probdist_{\loc_i}$ and we put
	$\loc_{i+1}=\loc'$;
	\item If $\loc_i$ is deterministic, then there is exactly one transition 
	$(\loc_i,\loc')$ enabled in $(\loc_i,\vec{x}_i)$, in which case we put 
	$\loc_{i+1}=\loc'$.
\end{itemize}
\end{compactitem}

The above intuitive explanation can be formalized by showing that
each pCFG $\pCFG$ together with a scheduler $\sigma$ and initial valuation 
$\vecinit\in \vecinitset$ uniquely determine a 
certain 
probabilistic 
space $(\OmegaRun,\natfilt,\probm^{\sigma}_{\vecinit})$ in which $\OmegaRun$ 
is a set of 
all 
runs in $\pCFG$, and a stochastic process 
$\pCFG^{\sigma}=\{\cfg{\sigma}{i}\}_{i=0}^{\infty}$ in this space 
such 
that for each run $\run\in \OmegaRun$ we have that $\cfg{\sigma}{i}(\run)$ is 
the 
$i$-th configuration on $\run$ (i.e., $\cfg{\sigma}{i}$ is a random vector 
$(\lr{\sigma}{i},\vr{\sigma}{i})$ with  $\lr{\sigma}{i}$ taking 
values in $\locs$ and $\vr{\sigma}{i}$ being a random vector of dimension 
$|\pvars|$ consisting of real-valued random variables). The sigma-algebra 
$\natfilt$ is the smallest (w.r.t. inclusion) sigma-algebra under which all 
the functions $\cfg{\sigma}{i}$, for all $i\geq 0$, 
are $\natfilt$-measurable (i.e., for each $\cfg{\sigma}{i}$ and each measurable 
set of configurations $\confset\in\mathcal{F}_{\Conf}$ it holds 
$\{\run\mid\cfg{\sigma}{i}(\run)\in \confset\}\in \natfilt$). Equivalently, 
$\natfilt$ can be defined as a sigma algebra generated by all set of runs of 
the form $F\times \Conf^{\infty}$, where $F\in \mathcal{H}$ is a measurable set 
of finite paths. The 
probability 
measure 
$\probm^{\sigma}_{\vecinit}$ is such that for each $i$, the distribution of 
$\cfg{\sigma}{i}$ reflects the aforementioned way in which runs are randomly 
generated. The formal construction 
of $\natfilt$ and $\probm^{\sigma}_{\vecinit}$ proceeds via the standard 
\emph{cylinder 
construction}~\cite[Theorem 2.7.2]{Ash:book} and is
somewhat 
technical, hence we omit it. 
We denote by $\E^\sigma_{\vecinit}$ the expectation operator in probability 
space 
$(\OmegaRun,\natfilt,\probm^{\sigma}_{\vecinit})$.

\vspace{-1em}
\subsection{Almost-Sure and Positive Termination}
\vspace{-0.5em}
Termination is the basic liveness property of \PP{}s. 

\smallskip\noindent{\em Termination and termination time.}
In the following, consider a \PP{} $P$ and its associated pCFG $\pCFG_P$. This 
pCFG has a special location  $\locterm$ corresponding to the value of 
the program counter after 
executing $P$.
We say that a run 
\emph{terminates} if it reaches a 
configuration whose first component is
$\loc_{P}^{\lout}$.
We define a random variable $\ttime$
such that for each run $\run$ the value $\ttime(\run)$ represents 
the first 
point in time when 
the current location is $\loc_P^{\lout}$. If a run $\run$ does 
\emph{not} 
terminate, then 
$\ttime(\run)=\infty$. We call $\ttime$ the 
\emph{termination time} of $\program$. 
Since a probabilistic program may exhibit more than one run, we are interested 
in probabilities of runs that terminate or reach some set of configurations. 
This gives rise to the following fundamental computational problems regarding 
termination:

\begin{compactenum}
\item \emph{Almost-sure termination:} A probabilistic program $P$ is 
almost-surely (a.s.) 
terminating if under each scheduler $\sigma$ and for each initial valuation 
$\vecinit\in\vecinitset$ it holds that 
$\probm^{\sigma}_{\vecinit}(\{\run\mid \run \text{ terminates}\}) = 1$, or 
equivalently, 
if for each $\sigma$ it holds $\probm^{\sigma}(\ttime<\infty)=1$. In 
almost-sure termination question for $P$ we aim to prove that $P$ 
is 
almost-surely terminating.
\item \emph{Positive termination:} A probabilistic program $P$ is 
positively
terminating if under each scheduler $\sigma$ and for each initial valuation 
$\vecinit\in\vecinitset$ it holds that 
$\E^{\sigma}_{\vecinit}(\{\ttime\})<  \infty$. In 
positive termination question for $P$ we aim to prove that $P$ 
is 
positively terminating. Note that each positively terminating program is also 
a.s. terminating, but the converse does not hold.
\end{compactenum}

\lstset{language=affprob}
\lstset{tabsize=3}
\newsavebox{\invrun}
\begin{lrbox}{\invrun}
\begin{lstlisting}[mathescape]
$x:=10$
while $x\geq 1$ do
	if prob(0.75) then $x:=x-1$	else $x:=x+1$ 
	fi
od
\end{lstlisting}
\end{lrbox}
\begin{figure}[t]
\centering
\usebox{\invrun}

\begin{tikzpicture}[x = 1.8cm]
%

\node[det] (while) at (1.5,0)  {$\loc_0$};
\node[ran] (prob) at (3,0) {$\loc_1$};
\node[det] (fin) at (0,0) {$\loc_2$};
%
\draw[tran] (while) to node[font=\scriptsize,draw, fill=white, 
rectangle,pos=0.5] {$x<1$} (fin);
\draw[tran, loop, looseness = 5, in =-65, out = -115] (fin) to (fin);
\draw[tran] (while) to node[font=\scriptsize,draw, fill=white, 
rectangle,pos=0.5] {$x\geq 1$} (prob);


\node (dum1) at (0,0.8) {};
\node (dum2) at (0,-0.8) {};

\draw[tran] (prob) -- node[font=\scriptsize,draw, fill=white, 
rectangle,pos=0.5, inner sep = 1pt] {$\frac{3}{4}$} (prob|-dum1) -- node[auto] 
{x:=x-1}
(while|-dum1)--(while);
\draw[tran] (prob) -- node[font=\scriptsize,draw, fill=white, 
rectangle,pos=0.5, inner sep = 1pt] {$\frac{1}{4}$} (prob|-dum2) 
-- node[auto,swap] {x:=x+1} (while|-dum2)--(while);
\end{tikzpicture}
\caption{An \APP{} modelling an asymmetric 1-D random walk and the associated 
pCFG. Probabilistic locations are depicted by circles, with probabilities given 
on outgoing 
transitions. Transitions are labelled by their effects. Location $\loc_0$ is 
initial and $\loc_2$ is terminal.}
\label{fig:invariant-running}
\end{figure}
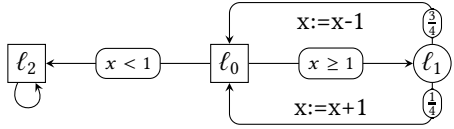

\section{Lexicographic Supermartingales}
\label{sec:lexicographic}

In this section we introduce the notion of a \emph{lexicographic ranking 
supermartingale}, which generalizes the standard notion of ranking 
supermartingales. However, to define any form of a supermartingale, we need the crucial notion of conditional expectation.

\smallskip\noindent{\bf Conditional Expectation.} 
Let $(\Omega,\mathcal{F},\probm)$ be a probability space, 
$X\colon\Omega\rightarrow 
\Rset$ an $\mathcal{F}$-measurable function, and $\mathcal{F}'\subseteq 
\mathcal{F}$ sub-sigma-algebra of $\mathcal{F}$. A \emph{conditional 
	expectation} of $X$ given $\mathcal{F}'$ is an $\mathcal{F}'$-measurable random 
variable denoted by $\E[X| \mathcal{F}']$ which satisfies, for each set $A\in 
\mathcal{F}'$, the following: 
\begin{equation}
\label{eq:cond-exp}
\E[X\cdot 1_A] = \E[\E[X|\mathcal{F}]\cdot 1_A],
\end{equation}
where $1_A \colon \Omega\rightarrow \{0,1\}$ is an \emph{indicator function} of 
$A$, i.e. function returning $1$ for 
each $\omega\in A$ and $0$ for each $\omega\in \Omega\setminus A$. Note that 
the left hand-side of~\eqref{eq:cond-exp} intuitively represents the expected 
value 
of $X(\omega)$ with domain restricted to $A$.

Note that any $\mathcal{F}'$-measurable random variable satisfying~\eqref{eq:cond-exp} can be called a conditional expectation. The definition does not guarantee that the conditional expectation is uniquely defined or that it exists at all. However, from probability theory we have the following:

\begin{proposition}
\label{prop:conditional-exp-existence}
\label{PROP:CONDITIONAL-EXP-EXISTENCE}
Let $(\Omega,\mathcal{F},\probm)$ be a probability space, 
$X\colon\Omega\rightarrow 
\Rset$ an $\mathcal{F}$-measurable function, and $\mathcal{F}'\subseteq 
\mathcal{F}$ sub-sigma-algebra of $\mathcal{F}$. Assume that one of the following conditions hold:
\begin{itemize}
\item $\E[|X|]<\infty$; or
\item $X$ is real-valued and non-negative.
\end{itemize}
Then there exists a conditional expectation of $X$ given $\mathcal{F}'$ and it is almost-surely unique, i.e. for each two $\mathcal{F}'$-measurable functions $f$, $g$ that satisfy~\eqref{eq:cond-exp} it holds $\probm(\{\omega\mid f(\omega)\neq g(\omega)\})=0$.
\end{proposition}
\begin{proof}[Proof (Key ideas)]
The proof for the case when $\E[|X|]$ is standard and appears in many textbooks on probability theory (e.g.~\cite{Billingsley:book,Ash:book,Rosenthal:book}). The proof for the second case is essentially the same: the condition that $X$ is non-negative and not admitting infinite value suffices for satisfying the assumptions of Radon-Nikodym Theorem, the main theoretical tool used in the proof. For the sake of completeness we present the proof in \AppendixMaterial.
\end{proof}

Since the constraint~\eqref{eq:cond-exp} defining conditional expectation is 
phrased in terms of expected values, the almost-sure uniqueness cannot be 
strengthened to uniqueness, as re-defining a random variable on a set of 
probability zero does not change its expectation. In the following, when we say 
that a conditional expectation of a random variable $X$ satisfies some 
inequality (e.g. $\E[X\mid\mathcal{F}]\geq 0$) on set $L\subseteq \Omega$, we 
mean that for each $\mathcal{F}$-measurable function $\E[X\mid\mathcal{F}]$ 
satisfying~\eqref{eq:cond-exp} the inequality holds on some subset $L'\subseteq 
L$ such that $\probm(L')=\probm(L)$.

In context of probabilistic programs we work with probability 
spaces of the form $(\Omega,\natfilt,\probm^\sigma)$, where $\Omega$ is a 
set of runs in some $\pCFG$ and $\natfilt$ is (the smallest) sigma-algebra 
such that all the functions $\cfg{\sigma}{i}$, where $i\in \Nset_0$ and 
$\sigma$ is a scheduler, are $\natfilt$-measurable. In such a setting we can 
also consider sub-sigma-algebras $\natfilt_i$, $i\in \Nset_0$, of 
$\natfilt$, where $\natfilt_i$ is the smallest sub-sigma-algebra of 
$\natfilt$ such that all the functions $\cfg{\sigma}{j}$, $0\leq j \leq 
i$, are $\natfilt_i$-measurable. Intuitively, each set $A$ belonging to such 
an $\natfilt_i$ consists of runs whose first $i$ steps satisfy some 
property, and the probability space $(\Omega,\natfilt_i,\probm^\sigma)$ 
allows us to reason about probabilities of certain events happening in the 
first 
$i$ steps of program execution. 
Then, for each $A\in \natfilt_i$, the 
value $\E[\E[X|\natfilt_i]\cdot 1_A]$ represents the expected value of 
$X(\run)$ for the randomly generated run $\run$ provided that we restrict to 
runs whose
prefix of length $i$ satisfies the property given by $A$. 
The sequence $\natfilt_0,\natfilt_1,\natfilt_2,\dots$ forms a 
filtration of $\natfilt$, which we call a \emph{canonical filtration}.


\begin{definition}[Lexicographic Ranking Supermartingale]
\label{def:lexrsm}
Let $(\Omega,\genfilt,\probm)$ be a probability space, 
$\{\genfilt_i\}_{i=0}^{\infty}$ a filtration of $\genfilt$, $\stime$ a stopping 
time w.r.t. that filtration, and 
$\eps\geq 0$. 
An $n$-dimensional real-valued stochastic process 
$\{\vec{X}_{i}\}_{i=0}^{\infty}$ is a 
\emph{lexicographic $\eps$-ranking supermartingale for $\stime$ 
($\eps$-LexRSM)} if the 
following 
conditions hold:
\begin{compactenum}
\item For each $1\leq j \leq n$ the 1-dimensional stochastic process 
$\{\vecseq{X}{i}{j}\}_{i=0}^{\infty}$ is adapted to 
$\{\genfilt_i\}_{i=0}^{\infty}$.
\item For each $\omega \in \Omega$, $i\in \Nset_0$ and $1\leq j \leq n$ it holds 
$\vec{X}_i (\omega)[j]\geq {0}$. 
\item For each $i\in \Nset_0$ there exists a partition of the set $\{\stime>i\}$ into $n+1$ subsets $L^i_1,\dots,L^i_n,L^i_{n+1}$, all of them $\genfilt_i$-measurable, such that for each $1\leq j \leq n$:
\begin{compactitem}
	\item $\E[\vecseq{X}{i+1}{j}\mid 
	\genfilt_i]\leq\vecseq{X}{i}{j} - 
	\eps$ on $L_j^i$; 
	\item 
	 for all $1 \leq j' < j$ we have $\E[\vecseq{X}{i+1}{j'}\mid 
	 \genfilt_i]=\vecseq{X}{i}{j'}$ on $L_{j}^i$; and
	\item $\E[\vecseq{X}{i+1}{j}\mid 
	\genfilt_i]\leq\vecseq{X}{i}{j}$ on $L_{n+1}^i$. 
\end{compactitem}
\end{compactenum}
The $n$-dimensional LexRSM is \emph{strict} if $L^i_{n+1}=\emptyset$ for each $i$.
\end{definition}

An \emph{instance} of an $n$-dimensional LexRSM $\{\vec{X} \}_{i=0}^{\infty}$ is a tuple $(\{\vec{X}_{i=0}^{\infty},\{L_1^i,\dots,L_{n+1}^i\}_{i=0}^{\infty})$ where the second component is a sequence of partitions of $\Omega$ satisfying the condition in Definition~\ref{def:lexrsm}.
Intuitively, the sets $ L_j^{i}$ for $1\leq j \leq n$ represent the lexicographic ranking condition, i.e. for strict LexRSMs, we are in each step able to partition $\Omega$ into subsets such that on $j$-th subset the $j$-th component of $\vec{X}$ is expected to decrease while the previous components are not expected to increase. On additional sets $L_{n+1}^i$, none of the components is expected to increase, but decrease is not required: this will become handy later when we deal with compositional LexRSM-based proofs (Section~\ref{sec:compositional}) -- we will not require decrease in every step as long as decrease happens at least once on each cycle in the pCFG. We say that $\omega\in \Omega$ has level $j$ in step $i$ of instance $(\{\vec{X}_{i=0}^{\infty},\{L_1^i,\dots,L_{n+1}^i\}_{i=0}^{\infty})$ if $\omega\in L^i_j$. We also say that $\omega$ has level $0$ in step $i$ if $\stime(\omega)\leq i$.

The strict 1-dimensional lexicographic $\eps$-ranking supermartingale is, to a large extent, equivalent to the notion of a ranking supermartingale as studied in~\cite{HolgerPOPL,CFNH16:prob-termination}. There is one significant difference: in these works there is an additional \emph{integrability} condition imposed on the one-dimensional process $\{X_i\}_{i=0}^{\infty}$, which requires that for each $i\geq 0$ it holds $\E[|X_i|]<\infty$ (or equivalently $\E[X_i]<\infty$, as the process is required to be non-negative). We do not impose this condition, which simplifies possible application of LexRSMs to programs with non-linear arithmetic, where, as already shown in~\cite{HolgerPOPL}, integrability of program variables is not guaranteed.
The reason why integrability condition can be dropped is that it is only needed in the previous works to ensure that the conditional expectations exist and are well-defined. However, the existence of conditional expectations is also guaranteed for random variables that are real-valued non-negative, see Proposition~\ref{prop:conditional-exp-existence}. 

The following theorem states our main mathematical result on LexRSMs.

\begin{theorem}
\label{thm:lexrsm-main}
\label{THM:LEXRSM-MAIN}
Let $(\Omega,\genfilt,\probm)$ be a probability space, 
$\{\genfilt_i\}_{i=0}^{\infty}$ a filtration of $\genfilt$, $\stime$ a 
stopping 
time w.r.t. that filtration, and $\eps>0$. Assume there exists an $n$-dimensional $\eps$-LexRSM for $\stime$ and its instance $(\{\vec{X}_{i=0}^{\infty},\{L_1^i,\dots,L_{n+1}^i\}_{i=0}^{\infty})$ such that $\probm(\{\omega\in \Omega \mid \text{level of $\omega$ is $<n+1$ in infinitely many steps}\})=1$. Then $\probm(\stime<\infty)=1$. In particular, if there exists a strict $\eps$-LexRSM for $\stime$, then $\probm(\stime<\infty)=1$.
\end{theorem}
\begin{proof}
The proof proceeds by contradiction, i.e. we assume that an $\eps$-LexRSM for 
$\stime$ satisfying the above conditions exists and $\probm(\stime=\infty)>0$. For succinctness we denote the 
set $\{\omega\mid \stime(\omega)=\infty\}$ by $\genRunSet_{\infty}$.

For $\omega\in \Omega$ we denote the level of $\omega$ at step $i$ by $\levelrank{\omega}{i}$.
The value $\levelrank{\omega}{i}$ is well-defined for all 
$\omega$ and moreover, the random variable $\levelrank{}{i}$ is $\genfilt_i$-measurable. We denote by $\minlev(\omega)$ the smallest $0 \leq j \leq n$ such 
that $j$ is a level of $\omega$ at infinitely many steps. Note that $\omega \in 
\genRunSet_{\infty}$ if and only if $\minlev(\omega)\neq 0$, so $\probm(\genRunSet_{\infty})=\probm(\{\minlev \neq 0 \})$. We denote by $M_i$ 
the set of all $\omega$'s with $\minlev(\omega)=i$.

Throughout the proof we use several times the following fundamental fact: if 
$\probm(A)>0$ for some set $A$ and $A=A_1\cup A_2 \cup A_3\cdots$ for some 
sequence of sets $A_1,A_2,A_3,\dots$, then there exists $i$ such that 
$\probm(A_i)>0$.

Now $\genRunSet_{\infty}=M_1\cup\dots\cup M_n\cup M_{n+1}$ and $\probm(M_{n+1})=0$ (as the measure of $\omega$'s that have level $<n+1$ in only finitely many steps is zero, per Theorem's assumption), there must be $1\leq \fixn{j} \leq n$ 
s.t. $\probm(M_{\fixn{j}})>0$, i.e. with positive probability the smallest level appearing infinitely often is $\fixn{j}$. For each $\omega\in M_{\fixn{j}}$ there is the smallest number
${i}_{\omega,\fixn{j}}\in\Nset_0$ such that for all $i\geq {i}_{\omega,\fixn{j}}$ it holds 
$\levelrank{\omega}{i}\geq \fixn{j}$, i.e. after step ${i}_{\omega,\fixn{j}}$ the level of $\omega$ in all steps up to infinity is at least $\fixn{j}$. Denote by $S_{\fixn{j},{i}}$ the set of all $\omega$'s in $M_{\fixn{j}}$ s.t. 
$i_{\omega,\fixn{j}}={i}$. 
Since $M_{\fixn{j}} = M_{\fixn{j},1} \cup M_{\fixn{j},2} \cup M_{\fixn{j},3} \cup \cdots$, there is $\fixn{i}\in 
\Nset_0$ 
s.t. 
$\probm(M_{\fixn{j},\fixn{i}})>0$. That is, there is a point in time such that with positive probability, after this point the level of $\omega$ is at least $\fixn{j}$, and it is equal to $\fixn{j}$ infinitely many times. Continuing on the same note, for each $B\in \Nset$ we 
denote by $M_{\fixn{j},\fixn{i}}^{B}$ the set off all $\omega$'s in $M_{\fixn{j},\fixn{i}}$ s.t. 
$\vecseq{X}{\fixn{i}}{\fixn{j}}(\omega)\leq B$. Since $M_{\fixn{j},\fixn{i}}= M_{\fixn{j},\fixn{i}}^{1} \cup M_{\fixn{j},\fixn{i}}^{2} 
\cup M_{\fixn{j},\fixn{i}}^{3} \cup \cdots  $, there is $\fixn{B}\in \Nset$ s.t. 
$\probm(M_{\fixn{j},\fixn{i}}^{\fixn{B}})>0$. 

So there is a set of positive probability $M=M_{\fixn{j},\fixn{i}}^{\fixn{B}}$  such that for all $\omega$'s in the set: after step $\fixn{i}$ the level of $\omega$ is at least $\fixn{j}$ (which, intuitively, means that $\vec{X}[\fixn{j}]$ does not have a tendency to increase after this time step on $\omega$'s in $M$), the level of $\omega$ is infinitely often equal to $\fixn{j}$ (intuitively, $\vec{X}[\fixn{j}]$ has infinitely often the tendency to decrease by $\geq \eps$ for $\omega$'s in $M$), and at time $\fixn{i}$ the value of $\vec{X}[\fixn{j}]$ is bounded (by $B$) on $M$. This should, again intuitively, lead to a conclusion, that when ``restricted to $M$'', $\vec{X}[\fixn{j}]$ has tendency to decrease unboundedly over time, a contradiction with non-negativeness of $\vec{X}[\fixn{j}]$. However, proving this intuitive result is much more intricate: most importantly, it is not clear what ``restricted to $M$'' stands for. The stochastic process $\{\vec{X}\}_{i=0}^{\infty}$ as well as the LexRSM conditions are tied to the filtration $\{\genfilt_i\}_{i=0}^{\infty}$, but the set $M$ is not necessarily $\genfilt_i$ measurable for any concrete $i$, since whether $\omega$ belongs to $M$ depends on values of $\levelrank{\omega}{i}$ for \emph{infinitely many $i$}. Hence, we use a work-around.

Let $D$ be the set of all $\omega\in \Omega$ such that 
$\vecseq{X}{\fixn{i}}{\fixn{j}}(\omega)\leq \fixn{B}$. Note that $M \subseteq D$ and $D\in 
\genfilt_{\fixn{i}}$ (and thus also $D\in \genfilt_{i'}$ for all $i'\geq \fixn{i}$). 
Define a stopping time $F$ w.r.t. filtration $\{\genfilt_i\}_{i=0}^{\infty}$ as follows: for all $\omega \in \Omega$ 
we put 
$F(\omega)=\inf\{k\in \Nset_0 \mid k\geq \fixn{i} \text{ and }
\levelrank{\omega}{k} < \fixn{j} \}$.

Define a 
(one-dimensional) stochastic process $\{Y_k\}_{k=0}^{\infty}$ as 
follows:
$$
Y_k(\omega) = \begin{dcases}
0 & \text{if $\omega \not\in D$}\\
\fixn{B} & \text{if $\omega \in D$ and $k<\fixn{i}$ }\\
\vecseq{X}{k}{\fixn{j}}(\omega) & \text{if $\omega \in D$, $k\geq \fixn{i}$ and 
$F(\omega)>k$ 
}\\
\vecseq{X}{F(\omega)}{\fixn{j}}(\omega) & \text{if $\omega\in D$, $k\geq 
\fixn{i}$ and 
$F(\omega)\leq k$ }.
\end{dcases}
$$

Intuitively, the process $\{Y_k\}_{k=0}^{\infty}$ is an over-approximation of what we would like to call ``$\vec{X}[\fixn{j}]$ restricted to $M$.''
We prove several properties of the process. 
First, clearly for all $k\geq 0$, $Y_k(\omega)\geq 0$. Second, for each $k\geq \fixn{i}$, the variable $Y_k$ is $\genfilt_k$-measurable, as $D\in \genfilt_{\fixn{i}}$, $\{\vecseq{X}{i}{\fixn{j}}(\omega)\}_{i=0}^{\infty}$ is adapted to the filtration $\{\genfilt_i\}_{i=0}^{\infty}$ and $F$ is a stopping time w.r.t. this filtration. Finally, for any $k\in \Nset_0$ denote by 
$\noofdec_k$ the random variable such that $\noofdec_k(\omega) = |\{i'\in 
\Nset\mid \fixn{i} \leq i' < k \text{ and } \levelrank{\omega}{i'}=\fixn{j}\} |$, i.e. $\noofdec_k(\omega)$ counts the number of steps between $\fixn{i}$ and $k$ in which level is $\fixn{j}$.
We prove
that for each $k\geq \fixn{i}$ 
it holds 
\begin{equation}
\label{eq:lexrsm-soundness-main}
\E[Y_k]\leq \fixn{B}\cdot \probm(D) - \eps\cdot\sum_{\ell=0}^{k-\fixn{i}} \ell\cdot\probm(D 
\cap \{F\geq k\} \cap \{\noofdec_k 
= 
\ell\}).
\end{equation}

The proof of~\eqref{eq:lexrsm-soundness-main} goes by induction on $k$. The computations being somewhat technical, we defer them to \AppendixMaterial. 

Now according to~\eqref{eq:lexrsm-soundness-main} it holds  $\E[Y_k]\leq \fixn{B}\cdot \probm(D) - \eps\cdot\sum_{\ell=0}^{k-\fixn{i}} \ell\cdot\probm(D 
\cap \{F\geq k\} \cap \{\noofdec_k 
= 
\ell\})$ for all $k\geq \fixn{i}$. Let $m=3\fixn{B}\cdot \probm(D)/(\eps\cdot \probm(M))$. For each $\omega\in M$ we see level $\fixn{j}$ infinitely often, so there exists step $k(\omega)\geq \fixn{i}$ such that $\noofdec_{k(\omega)}\geq m$, i.e. $\omega$ has level $\fixn{j} $ at least in $m$ steps between steps $\fixn{i}$ and $k(\omega)$. Clearly, $M = \bigcup_{\ell=\fixn{i}}^{\infty} (M \cap \{k(\omega) \leq \ell\})$ and hence $\probm(M)=\lim_{\ell\rightarrow \infty}\probm(M \cap \{k_{m} \leq \ell\})$. Thus, there exists $\ell_0 \geq m+\fixn{i}$ such that $\probm(M \cap \{\noofdec_{\ell_0}\geq m\}) \geq \probm(M)/2$. Clearly $M \cap \{\noofdec_{\ell_0}\geq m\} \subseteq D\cap \{F\geq \ell_0 \} \cap \{\noofdec_{\ell_0}\geq m \}$. From \eqref{eq:lexrsm-soundness-main} it follows that
\begin{align*}
\E[Y_{\ell_0}]&\leq \fixn{B}\cdot \probm(D) - \eps\cdot\sum_{\ell=0}^{\ell_0-\fixn{i}} \ell\cdot\probm(D 
\cap \{F\geq \ell_0\} \cap \{\noofdec_{\ell_0} 
=
\ell\}) \\
&= \fixn{B}\cdot \probm(D) - \eps\cdot \sum_{\ell=1}^{\ell_0-\fixn{i}} \probm(D 
\cap \{F\geq \ell_0\} \cap \{\noofdec_{\ell_0} 
\geq 
\ell\})\\
& \leq \fixn{B}\cdot \probm(D) - \eps\cdot \sum_{\ell=1}^{m} \probm(D 
\cap \{F\geq \ell_0\} \cap \{\noofdec_{\ell_0} 
\geq 
\ell\}) \\&\leq \fixn{B}\cdot \probm(D) - \eps\cdot \sum_{\ell=1}^{m} \probm(D 
\cap \{F\geq \ell_0\} \cap \{\noofdec_{\ell_0} 
\geq 
m\})\\
&\leq \fixn{B}\cdot \probm(D) - \eps\cdot m \cdot \probm(M)/2 < 0, 
\end{align*}where the second line follows by standard re-arranging of terms, the third line follows from the fact that $m\leq \ell_0 + \fixn{i}$, the fourth line follows from $\{\noofdec_{\ell_0} 
\geq 
m \} \subseteq\{ \noofdec_{\ell_0} 
\geq 
\ell\}$ for each $\ell \leq m$, the first inequality on the last line follows 
by  using $\probm(M)/2\leq \probm(M \cap \{\noofdec_{\ell_0}\geq m\}) \leq 
\probm(D\cap \{F\geq \ell_0 \} \cap \{\noofdec_{\ell_0}\geq m \})$, and
the last inequality follows by expanding the definition of $m$. But for each $k$ the random variable $Y_k$ is non-negative, so it must also have a non-negative expectation, a contradiction. Finally, note that for strict $n$-dimensional LexRSMs the condition of level $<n+1$ appearing infinitely many times is trivially satisfied.
\end{proof}

\section{Applying Lexicographic Supermartingales to Probabilistic Programs}
\label{sec:lex-programs}

We now discuss how to leverage the mathematical results of the previous section 
to provide a sound proof rule for almost-sure termination of probabilistic 
programs. Hence, for the rest of this section we fix a \PP{} $\program$ and the 
associated pCFG 
$\pCFG_\program=(\locs,\pvars,\locinit,\vecinitset,\transitions,\updates,\probdist,\guards)$.

We aim to define a function assigning a non-negative vector to each 
configuration (so called measurable map) such that in each point of 
computation, the expected value of the function after performing one more 
computational step is smaller (in lexicographic ordering) than the current one. 
We formalize this property below.

\begin{definition}[Measurable Maps and Linear Expression Maps]
A 1-dimensional \emph{measurable map} for a \PP{} $\program$ is a  
real-valued function $\lem$ 
assigning to each program location $\loc$ of $\pCFG_P$ a Borel-measurable function $\lem(\loc)$  of program variables, i.e. each $\lem(\loc)$  is a function of type $\Rset^{|\pvars|}\rightarrow \Rset$. As a special case, if all the functions $\lem(\loc)$ are affine, then we call $\lem$ a 1-dimensional \emph{linear expression map (LEM)}. 
Am $n$-dimensional measurable/linear expression map is a vector $\vec{\lem}=(\lem_1,\dots,\lem_n)$ of 1-dimensional measurable/linear expression maps. 
\end{definition}

Each 1-dimensional measurable map $\lem$ and location $\loc$ determines a function $\lem(\loc)$ 
which takes as an argument an $|\pvars|$-dimensional vector. We use $\lem(\loc,\vec{x})$ as a shorthand 
notation for $\lem(\loc)(\vec{x})$.

We now formalize the notion of a transition in a pCFG being ranked by a 
measurable map. We first define this notion for transitions that do not go out 
of a probabilistic branching location, as these require a special treatment.

\begin{definition}
\label{def:rank1}
Let $\lem$ be a measurable map, $(\loc,\vec{x})$ be a configuration such that 
$\loc\not\in \locsPB$ and let 
$\tau=(\loc,\loc')$ be a 
transition outgoing from $\loc$. For an $\eps\geq 0$ we say that $\tau$ is 
\emph{$\eps$-ranked by $\lem$ from $(\loc,\vec{x})$} if the following 
conditions are 
satisfied, depending on the type of $\loc$:
\begin{compactitem} 
		\item 
		if $\loc$ is a deterministic or non-deterministic branching location, 
		then
		$$\lem(\loc',\vec{x}) \leq 
		\lem(\loc,\vec{x})-\eps;$$
		
		
		\item 
		If $\loc$ is an assignment location, then we distinguish 
		three cases, depending on $\updates(\tau)=(j,\up)$ (recall that $\up$ 
is an update element):
		\begin{compactitem}
			\item If $\up\colon \Rset^{|\pvars|}\rightarrow \Rset$ is a 
Borel-measurable function, then we require
			$$\lem(\loc',\vec{x}(j\leftarrow 
\up(\vec{x}))) \leq \lem(\loc,\vec{x})-\eps$$
			\item If $\up$ is a distribution $d$, then we require $$ 
			\lem(\loc',\vec{x}(j\leftarrow \E[d])) \leq 
			\lem(\loc,\vec{x})-\eps,$$ 
			where $\E[d]$ is the 
expected value of the 
			distribution $d$.
			\item 
			If $\up$ is a set, then we require $$\sup_{a\in\up}
			\lem(\loc',\vec{x}(j \leftarrow a)) \leq \lem(\loc,\vec{x})-\eps.$$
		\end{compactitem}
		
	\end{compactitem}
\end{definition}

Since ranking supermartingales are required to decrease on average, for 
individual transitions outgoing from $\locsPB$ it does not make sense to say 
that they are ranked or not. Instead, for each $\loc\in\locsPB$ we consider all 
outgoing transitions together.

\begin{definition}
\label{def:rank2}
Let $\lem$ be a measurable map, and let $(\loc,\vec{x})$ be a configuration with
$\loc\in \locsPB$. For an $\eps\geq 0$ we say that $\loc$ is $\eps$-ranked by 
$\lem$ from $(\loc,\vec{x})$ if 
$$\sum_{(\loc,\loc')\in\transitions} 
		Pr_{\loc}\left(\loc,\loc'\right)\cdot
		\lem(\loc',\vec{x}\leq \lem(\loc,\vec{x})-\eps.$$
\end{definition}

To capture the specific of $\locsPB$, we introduce the 
notion of \emph{generalized transition}.

\begin{definition}
A generalized transition of a pCFG $\pCFG$ is either a transition of $\pCFG$ 
outgoing 
from location not in $\locsPB$ or a location $\loc\in\locsPB$.
\end{definition}

Intuitively, we represent the set of transitions outgoing from $\loc\in\locsPB$ 
by the source location $\loc$.  For generalized transitions 
$\tilde{\tau}=\loc\in\locsPB$ we say that $\tilde{\tau}$ is outgoing from 
$\loc$.

Definitions~\ref{def:rank1} and~\ref{def:rank2} 
define when is a generalized transition $\eps$-ranked by $\lem$ from 
configuration $(\loc,\vec{x})$. We say that a 
generalized transition is \emph{unaffected} by $\lem$ from $(\loc,\vec{x})$ if 
it is $0$-ranked by $\lem$ from $(\loc,\vec{x})$.

As in termination analysis of non-probabilistic programs, our LexRSMs are typically supported by \emph{invariants}, i.e. overapproximations of the set of reachable configuration. 

\begin{definition}[Invariant Map and Linear Invariant Map]
An \emph{invariant map} for a \PP{} $\program$ is a function $\inv$ assigning to each location of $\pCFG_{\program}$ a Borel-measurable set $\inv({\loc})\subseteq \Rset^{|\pvars|}$ of variable valuations, so called invariant of $\loc$, such that for each configuration $(\loc,\vec{x})$ reachable from the initial configuration it holds $\vec{x}\in \inv(\loc)$. Additionally, if each set $\inv(\loc)$ is of the form $\{\vec{x}\mid\vec{x}\models \Psi^\ell \}$ for some propositionally linear predicate $\Psi^\ell$, then we call $\inv$ a \emph{linear invariant map} (LIM).
\end{definition}

Slightly abusing the notation, we view each LIM equivalently as a function assigning linear predicates (whose satisfaction sets overapproximate the set of reachable valuations) to program locations.

We now have all the ingredients needed to define the notion of LexRSM maps for probabilistic programs. For notational convenience, we extend the function $\guards$ (which assigns guards to deterministic transitions) to the set of all generalized transition: for a generalized transition $\tau'$ which is not a standard transition outgoing from deterministic location, we put $\guards(\tau')=0\leq 0 \equiv \mathit{true}$.

\begin{definition}[Lexicographic Ranking Supermartingale Map]
Let $\eps>0$. An $n$-dimensional \emph{lexicographic $\eps$-ranking supermartingale map} ($\eps$-LexRSM map) for a program $\program$ supported by an invariant map $\inv$ is an $n$-dimensional measurable map $\vec{\lem}=(\lem_1,\dots,\lem_n)$ for $\program$ such that for each configuration $(\loc,\vec{x})$ where $\loc\neq \locterm$ and $\vec{x}\in \inv(\loc)$ the following conditions are satisfied:
 \begin{compactitem}
 	\item
 	for all $1\leq j \leq n$, $\lem_j(\loc,\vec{x})\geq 0$; and
 	\item 
 	for each generalized transition $\tilde{\tau}$ outgoing from $\loc$ such that $\vec{x}\models\guards(\tilde\tau)$ there 
 	exists $1\leq j 
 	\leq$ n such that
 	\begin{compactitem}
 	\item
 	$\tilde{\tau}$ is $\eps$-ranked by $\lem_j$ from $(\loc,\vec{x})$
 	\item
 	for all $1\leq j'<j$ we have that $\tilde{\tau}$ is unaffected by 
 	$\lem_{j'}$ from $(\loc,\vec{x})$.
 	\end{compactitem}
 \end{compactitem}
If additionally $\lem$ is a linear expression map, then we call it a linear $\eps$-LexRSM map ($\eps$-LinLexRSM).
\end{definition}


The main result is the soundness of $\eps$-LexRSM maps for proving a.s. 
termination.

\begin{theorem}
\label{thm:lexrsm-programs}
Let $\program$ be a probabilistic program. Assume that there exists an $\eps>0$ 
and an $n$-dimensional $\eps$-LexRSM map $\vec{\lem}=(\lem_1,\dots,\lem_n)$ for 
$\program$ supported 
by some 
invariant map $\inv$. 
Then $\program$ terminates almost surely.
\end{theorem}
\begin{proof}
Let $\sigma$ be any measurable scheduler and $\vecinit\in\vecinitset$ any 
initial variable valuation in $\program$.
We define an $n$-dimensional stochastic process 
$\{\vec{X}_{i}\}_{i=0}^{\infty} $ on the probability space 
$(\OmegaRun,\natfilt,\probm^{\sigma}_{\vecinit})$ such 
that for each 
$i\geq 0$ and $1\leq j 
\leq n$ and each run $\run$ we put $\vecseq{X}{i}{j}(\run) = 
\lem_j(\cfg{\sigma}{i}(\run))$. We claim that $\{\vec{X}_{i}\}_{i=0}^{\infty}$ 
is a strict $n$-dimensional $\eps$-LexRSM for the termination time $\ttime$ of $\program$. Clearly 
the process is real-valued, componentwise non-negative, and adapted to the 
canonical filtration of $\natfilt$. It remains to prove that condition (3) in 
Definition~\ref{def:lexrsm} is satisfied. To this end, for each $i\geq 0$ we 
define an almost-sure partition of the set $\{\run\in \OmegaRun\mid 
\ttime(\run) >i\}$ into sets $L^{i}_1,\dots,L^{i}_n$ by putting $L^i_j$ to be 
the set of all runs $\run$ such that $\ttime(\run)>i$ and for $\run$ the index 
$j$ is the smallest one such that the $(i+1)$-th transition on $\run$ is ranked 
by $\lem_j$ from $\cfg{\sigma}{i}(\run)$. Due to 
definition of an $\eps$-LexRSM map such a $j$ exists for all 
$\run\in\{\ttime>i\}$ and hence we indeed have a partition (so $L^i_{n+1}=\emptyset$ for all $i$). It remains to prove 
that irrespective of the initial choice of $\sigma$ and $\vecinit$ it holds, 
for each $1\leq j 
\leq n$ and $j'<j$, that $\E^\sigma_{\vecinit}[\vecseq{X}{i+1}{j}\mid 
\natfilt_i]\leq X_i[j]-\eps $ on $L_j^i$ and 
$\E^\sigma_{\vecinit}[\vecseq{X}{i+1}{j}\mid 
\natfilt_i]\leq X_i[j] $ on $L_{j'}^i$. This follows easily from 
the definition of $L^{i}_1,\dots,L^{i}_n$ and from the definition of a 
transition being $\eps$-ranked by $\lem_j$.

Since  $\{\vec{X}_{i}\}_{i=0}^{\infty}$ 
is an $\eps$-LexRSM for $\ttime$, from Theorem~\ref{thm:lexrsm-programs} it 
follows that $\probm^{\sigma}_{\vecinit}(\ttime <\infty) =1$, irrespective of 
$\sigma$ and $\vecinit$.
\end{proof} 

We conclude this section by showing that (Lin)LexRSMs can, unlike 1-dimensional 
RSMs, prove a.s. 
termination of programs whose expected termination time is infinite.
\lstset{language=affprob}
\lstset{tabsize=2,escapechar=\&}
\newsavebox{\infas}
\begin{lrbox}{\infas}
	\begin{lstlisting}[mathescape]
$x:=1$;$c:=1$ 
while $c\geq 1$ do				
	if prob(0.5) then			
		$x:=2\cdot x$			
	else						
		$c:=0$					
	fi							
od								
while $x\geq 0$ do $x:=x-1$ od	
	\end{lstlisting}
\end{lrbox}
\newsavebox{\infast}
\begin{lrbox}{\infast}
	\begin{lstlisting}[mathescape]

$(6c+2,x)$
$(6c+1,x)$
$(6c+3,x)$

$(3,x)$


$(0,x)$	
	\end{lstlisting}
\end{lrbox}
\newsavebox{\infastinv}
\begin{lrbox}{\infastinv}
	\begin{lstlisting}[mathescape]

$[c\geq 0 \wedge x\geq 1]$
$[c\geq 1 \wedge x\geq 1]$
$[c\geq 1 \wedge x\geq 1]$

$[c\geq 1 \wedge x \geq 1]$


$[x\geq 0]$	
	\end{lstlisting}
\end{lrbox}
\begin{figure}[t]
	\centering
	\usebox{\infas}
	\hspace{0.1cm}
	\usebox{\infast}
	\hspace{0.1cm}
	\usebox{\infastinv}
\caption{An a.s. terminating program with infinite expected termination time. A 
2-dimensional $1$-LinLexRSM map for the program is given on the right, along 
with the 
supporting invariants in square brackets. The invariants and a LinLexRSM on 
each 
line belong to the program location in which the program is \emph{before} 
executing the command on that line. The function is indeed a $1$-LinLexRSM, 
since in the probabilistic branching location $\loc$ we have 
$\preexp{\lem}(\loc,(x,c))=3c+3$ and $3c+3\leq 6c+1-1$ for all $c\geq 1$.} 
\label{fig:inftime}
\end{figure}

\begin{example}
\label{ex:infinite-time}
Consider the program in Figure~\ref{fig:inftime}. It is easy to see that it 
terminates a.s., but the expected termination time is infinite: to see this, 
note that that the expected value of variable $x$ upon reaching the second loop 
is $\frac{1}{2}\cdot 1 + \frac{1}{4}\cdot 2 + \frac{1}{8}\cdot 4 + \cdots = 
\frac{1}{2}+\frac{1}{2}+\frac{1}{2}+\cdots=\infty$ and that the time needed to 
get out of the second loop is equal to the value of $x$ upon entering the loop. 
However, a.s. termination of the program is proved by a 2-dimensional LinLexRSM 
$\lem$ pictured in the figure.
\end{example}

\section{Algorithmic Aspects}
\label{sec:algo}

In this section we describe a polynomial-time algorithm for synthesizing linear 
$\eps$-LexRSM maps in affine probabilistic programs supported by a given linear 
invariant map $\inv$. The algorithm, based on iterative solving of linear 
constraints, is a generalization of an algorithm for finding 
lexicographic ranking functions in non-probabilistic 
programs~\cite{ADFG10:lexicographic}. Hence, we provide only a high-level 
description, focusing on the new aspects.

The main idea is to iteratively synthesize 1-dimensional linear expression map that $1$-rank a subset of generalized transitions. These maps form the individual components of the sought-after $1$-LinLexRSM map. In each iteration, we start with a set $U$ of the yet-unranked generalized transitions. We seek a 1-dimensional LEM which ranks the maximal number of elements if $U$, and is unaffected by the remaining elements of $U$ (here by ranking a generalized transition $\tilde\tau$ we mean ranking it from each configuration $(\loc,\vec{x})$, where $\loc$ is the source location of $\tilde\tau$ and $\vec{x}\in\inv(\loc)$). If no 1-dimensional LEM that would rank at least one element in $U$ exists, then there is no LinLexRSM map for the program. Otherwise, we remove the newly ranked elements from $U$ and continue into the next iteration, until $U$ becomes empty. The process is summarized in Algorithm~\ref{algo:linlexrsm}.

Hence, the main computational task of the algorithm is to check, for a given set of generalized transition $U$, whether there exists a 1-dimensional LEM $\lem$ such that:
\begin{compactenum}
\item for each location $\loc\in\locs$ and all $\vec{x}\in \inv(\loc)$ it holds $\lem(\loc,\vec{x})\geq 0$;
\item for each $\tilde{\tau}\in U$ and  each configuration $(\loc,\vec{x})$ where $\loc$ is the source of $\tilde{\tau}$ and $\vec{x}\in \inv(\loc)\cap\{\vec{x}'\mid \vec{x}'\models \guards(\tilde{\tau})\}$ we have that $\tilde{\tau}$ is unaffected by $\lem$ from $(\loc,\vec{x})$; and
\item there is $\tilde{\tau}\in U$ that is $1$-ranked by $\lem$, from each configuration $(\loc,\vec{x})$ where $\loc$ is the source of $\tilde{\tau}$ and $\vec{x}\in \inv(\loc)\cap\{\vec{x}'\mid \vec{x}'\models \guards(\tilde{\tau})\}$; we then say that $\lem$ ranks $\tilde{\tau}$ w.r.t. $\inv$.
\end{compactenum}
  Moreover, if such an LEM $\lem$ exists, the algorithm has to find one that 
  maximizes the number of gen. transitions in $U$ ranked by it. Both these 
  tasks can be accomplished by the standard  method of linear constraints based 
  on the use of Farkas's lemma, which was widely use for synthesis of 
  termination proofs in both probabilistic and non-probabilistic 
  programs~\cite{DBLP:conf/tacas/ColonS01,DBLP:conf/vmcai/PodelskiR04,SriramCAV,CFNH16:prob-termination}.
   That is, the algorithm first constructs, for each location $\loc$ a 
  \emph{template} for $\lem$, i.e. an expression of the form $a_1^{\loc}x_1 + 
  \cdots + a_{|\pvars|}^{\loc}x_{|\pvars|} = b^{\loc} $, where 
  $x_1,\dots,x_{|\pvars|}$ are program variables and 
  $a_1^{\loc},\dots,a_{|\pvars|}^{\loc},b^{\loc} $ are yet unknown 
  coefficients. That is, supplying concrete values for all the unknown 
  coefficients yields an LEM. Now the conditions (1) and (2) above can be 
  expressed using linear constraints on the coefficients. More precisely, using 
  the construction provided e.g. in~\cite{SriramCAV,CFNH16:prob-termination} 
  (which includes a use of the Farkas's lemma) we construct in polynomial time, 
  for each generalized transition $\tilde{\tau}$, a system of linear 
  constraints $\linsystem_{\tilde{\tau}}$ over set of variables $\{ 
  a_1^{\loc},\dots,a_{|\pvars|}^{\loc},b^{\loc}\mid \loc\in\locs \}\cup 
  \{\eps_{\tilde\tau}\}\cup F$, where $F$ is the set of fresh variables (not 
  appearing in any template) and $\eps_{\tilde\tau} $ is constrained to be 
  non-negative. Each solution of the system $\linsystem_{\tilde{\tau}}$ yields 
  a LEM which satisfies the constraints (1) and (2) for $\tilde{\tau}$. 
  Moreover, each solution of $\linsystem_{\tilde{\tau}}$ yields a LEM which 
  $\eps_{\tilde{\tau}}$-ranks $\tilde{\tau}$. To find a LEM which satisfies all 
  constraints (1)--(3) as well as maximizes the number of 1-ranked elements of 
  $U$ it is sufficient to construct $\linsystem_{\tilde{\tau}}$ for each 
  $\tilde{\tau}\in U$ and solve the following linear program $\lp_{U}$:
  \begin{align*}
\text{maximize }  &\sum_{\tilde{\tau} \in U} \eps_{\tilde\tau} \text{ subject to constraints}\\[0.3cm]
&\linsystem_{\tilde{\tau}}\,;\quad\quad \quad\quad\tilde{\tau}\in U\\
&0 \leq \eps_{\tilde{\tau}} \leq 1\,; \quad \tilde{\tau}\in U
  \end{align*}
  
Each system $\linsystem_{\tilde{\tau}}$ is constructed in such a way that if it admits a solution with some $\eps_{\tilde{\tau}}$ positive, then decreasing the value of $\eps_{\tilde{\tau}}$ in that solution to any non-negative value still yield a valid solution (this corresponds to the fact that if some transition is $\eps$-ranked by $\lem$, than it is $\eps'$-ranked by $\lem$ for each $0\leq \eps'\leq \eps$). Morever, each solution where $\eps_{\tilde{\tau}}$ is positive can be rescaled into another solution in which $\eps_{\tilde{\tau}}$ is at least $1$. It follows that of $\lp_{U}$ has at least one feasible solution, then it has an optimal solution in which each $\eps_{\tilde{\tau}}$ is either $0$ or $1$. If the system does not have a solution or all the $\eps_{\tilde{\tau}}$ are equal to zero then there is no LEM satisfying (1)--(3). Otherwise, the optimal solution of $\lp_{U}$ yields a LEM $\lem$ which satisfies (1)--(3) and maximizes the number of 1-ranked elements of $U$.

This polynomial-time linear-programming step is used as a sub-procedure in Algorithm~\ref{algo:linlexrsm} for LinLexRSM synthesis. 

\begin{algorithm}
\SetKwInOut{Input}{input}\SetKwInOut{Output}{output}
\DontPrintSemicolon

\Input{An \APP{} $\program$ together with an invariant map $\inv$.}
\Output{A multi-dimensional LinLexRSM if it exists, otherwise ``No LinLexRSM''}
$U \leftarrow \text{ all generalized transitions of } \pCFG_{\program}$\;
$d\leftarrow 0$\;
\While{$U$ is non-empty}{
	$d \leftarrow d+1$\;
	construct and solve $\lp_{U}$\;
	\If{$\lp_U$ does not have feasible solution or optimal value is $0$}{\Return{No LinLexRSM}}
	\Else{$\sol\leftarrow$ optimal solution of $\lp_U$\;
		$\lem_d \leftarrow $ the LEM $\lem$ induced by $\sol$\;
		$U\leftarrow U \setminus \{\tilde{\tau}\mid \eps_{\tilde{\tau}}=1 \text{ in }\sol \}$}
}
\Return{$(\lem_1,\dots,\lem_d)$}
\caption{Synthesis of LinLexRSMs for \APP{}s}
\label{algo:linlexrsm}
\end{algorithm}

Both soundness and relative completeness of Algorithm~\ref{algo:linlexrsm} 
follow by arguments identical to those presented 
in~\cite{ADFG10:lexicographic}. 
\begin{theorem}
\label{thm:algo}
\label{THM:ALGO}
	Suppose that Algorithm~\ref{algo:linlexrsm} is run on an \APP{} $\program$ 
	together with linear invariant map $\inv$. If the algorithm returns a 
	$d$-dimensional LEM $\vec{\lem}=(\lem_1,\dots,\lem_d)$, then $\vec{\lem}$ 
	is a $1$-LinLexRSM map for $\program$ supported by $\inv$. Conversely, if 
	the algorithm returns ``No LinLexRSM'', then for any $d'\in \Nset$ and 
	$\eps>0$ there is no $d'$-dimensional $\eps$-LinLexRSM for $\program$  
	supported by $\inv$. If guards of all conditionals and loops in $\APP{}$ 
	are linear assertion, then the
	algorithm runs in time polynomial in size of 
	$\program$ and $\inv$.
\end{theorem}
\begin{proof}[Proof (Key Ideas)]
For soundness, let $U_i$ and $U_{i}$ denote the content of $U$ just before the $i$-th iteration of the while loop. Then each $\tilde{\tau}\in U_{i}\setminus U_{i+1}$ is $1$-ranked by $\lem_i$, for each $1\leq i \leq d$. Since $U_{d}=\emptyset$, each generalized transition of $\pCFG_{\program}$ is 1-ranked by some component of $\vec{\lem}$. Non-negativity of each $\lem_i$ is ensured directly by $\lp_{U_i}$. Hence, it remains to show that each $\tilde{\tau}\in U_i$ is unaffected by $\lem_1,\dots,\lem_{i-1}$. But this follows from $U_i \subseteq U_{i-1}\subseteq \cdots\subseteq U_1$ and from the fact that each gen. transition in $U_j$ is unaffected by $\lem_j$.

Proving completeness is more intricate; as pointed out 
in~\cite{ADFG10:lexicographic}, one needs to show that the greedy strategy of 
selecting LEM that ranks maximal number of remaining transitions does not cut 
off some possible LinLexRSMs. In~\cite{ADFG10:lexicographic} the completeness 
of the greedy strategy is proved by using several geometric arguments that 
exploit the fact that the underlying programs are affine. The same geometric 
properties hold for our generalization to \APP{}s (all ranking conditions in 
Definitions~\ref{def:rank1} and~\ref{def:rank2} are linear in program 
variables), so the result is easily transferable. The complexity argument is 
rather standard and we present it in~\AppendixMaterial.
\end{proof}

\section{Compositionality of Ranking Supermartingales Revisited}
\label{sec:compositional}

\subsection{One-Dimensional Compositional Proofs of Almost-Sure Termination}

Compositionality in the context of termination proving means providing the 
proof of termination step-by-step, handling one loop at a time, rather than 
attempting to construct the proof (in our case, a LexRSM) at 
once~\cite{KSTW10:compositional-transition-invariants}. 
In the context of 
probabilistic programs, the work~\cite{HolgerPOPL} attempted to provide a 
compositional notion of almost-sure termination proof based on the 
\emph{probabilistic variant rule (V-rule),} which we explain in a more detail 
below. However, for the method to work,~\cite{HolgerPOPL} imposes a 
technical~\emph{uniform integrability} condition, whose checking is hard to 
automatize. In this 
section we show that using our insights into LexRSMs we 
can obtain a different notion of a probabilistic V-rule which is sound 
without any additional assumptions, and which can be used to compositionally 
prove termination 
of programs that the previous method cannot handle.

Let $\program$ be a \PP{} of the form $\textbf{while } \Psi \textbf{ do } 
\programbody \textbf{ od}$, and let $\pCFG_{\program}$ be the associated pCFG, whose set of locations we denote by $\locs$. We denote by $\loops(\program)$ the set of all locations of $\pCFG_{\program}$ that belong to a sub-pCFG of $\pCFG_{\program}$ corresponding to some nested loop of $\program$. We also define $\slice(\program)$ to be the set $\locs\setminus\loops(\program)$ of locations that do not belong to any nested sub-loop.  A formal definition of 
both functions is given in \AppendixMaterial, we illustrate them in the 
following example.

\begin{example}
Consider the program $\program$ in Figure~\ref{fig:uniint4} and its associated 
pCFG. Then $\loops(\program)=\{\loc_2,\loc_3\}$ and 
$\slice(\program)=\{\loc_0,\loc_1,\loc_4,\loc^{\lout}\}$.
\end{example}

Given an invariant map $\inv$, we say that a 1-dimensional measurable map 
$\lem$ for $\pCFG_{\program}$ is $\eps$-$\inv$-ranking/unaffecting in location $\loc$, if 
for each $\vec{x}\in\inv(\loc)$  each generalized transition~$\tilde{\tau}$ 
outgoing from $\loc$ is $\eps$-ranked/unaffected by $\lem$. 

We recall the notion of compositional ranking supermartingale as introduced 
in~\cite{HolgerPOPL}. We call it a PV supermartingale, as it is based on so 
called probabilistic variant rule. Due to differences in syntax and semantics, 
the definition is syntactically slightly different from~\cite{HolgerPOPL}, but 
the essence is the same. A measurable map $\lem$ is propositionally linear, if 
each function $\lem(\loc)$ is of the form $\indicator{G_1}\cdot E_1 + \cdots 
\indicator{G_k}\cdot E_k$, where each $\indicator{G_i}$ is an indicator 
function of some polyhedron and each $E_i$ is a linear expression.

\begin{definition}[PV-supermartingale {\cite[Definition 7.1.]{HolgerPOPL}}]
A 1-dimensional propositionally linear map $\lem$ is a PV supermartingale (PVSM) for a program $\program$ supported by an invariant map $\inv$ if there exists $\eps>0$ such that $\lem$ is $\eps$-$\inv$-ranking an $\inv$-non-negative in each location $\loc\in\slice{(\program)}$ and $\inv$-unaffected in each $\loc\in\loops(\program)$.
\end{definition}

As a matter of fact, the condition that $\lem$ should be non-negative in locations of $\slice{(\program)}$ is not explicitly mentioned in \cite{HolgerPOPL}. However, it is implicitly used in some of the proofs and one can easily construct an example where, if the non-negativity in $\slice{(\program)}$ is not required, the Theorem~\ref{thm:holger-comp} below, which also comes from~\cite{HolgerPOPL}, does not hold. Hence, we state the condition explicitly.

In~\cite{HolgerPOPL} they show that even if all nested loops were already 
proved to terminate a.s. and there is a PVSM for the program, then the program 
itself might not terminate a.s. Then they impose a \emph{uniform integrability} 
constraint on the PVSM under which a 
PVSM together with a proof of a.s. termination of each nested sub-loop of 
$\program$ entails termination of the whole program $\program$. 
Uniform integrability is a deep 
concept from probability and measure theory: a sequence $X_0,X_1,X_2,\dots$ of 
random variables is uniformly integrable if for each $\delta>0$ there exists an 
$K\in \Nset$ such that for all $n\geq 0$ it holds $\E[|X_n|\cdot\indicator{X_n 
	\geq K}]\leq \delta$.
Apart from uniform 
integrability being somewhat restrictive in itself, in~\cite{HolgerPOPL} it is 
argued that proving uniform integrability is beyond the capability of 
state-of-the-art automated theorem provers. As a substitute for 
these,~\cite{HolgerPOPL} introduces a type system that can be used to 
automatically prove 
uniform integrability of ranking supermartingales for a restricted class of 
programs.
We do not repeat the precise definition of the typesystem here, we just say 
that a PVSM satisfying the condition imposed by the type system 
\emph{typechecks} correctly. 
In~\cite{HolgerPOPL} the following was proved:

\begin{theorem}[\cite{HolgerPOPL}]
\label{thm:holger-comp}
Let $\program$ be a \PP{} of the form $\textbf{while } \Psi \textbf{ do } 
\programbody \textbf{ od}$. Assume that each nested loop of $\program$ 
terminates almost surely from each reachable configuration, and that there 
exists a PVSM for $\program$ that typechecks correctly. Then $\program$ 
terminates almost surely.
\end{theorem}


The intricacies of uniform integrability are shown in the following example.

\begin{example}
\label{ex:uniform}
\label{EX:UNIFORM}
Consider the two \APP{}s in Figure~\ref{fig:uniint}, that differ only in one 
coefficient in the 
assignment on line 5. For the inner loop there exists (in both cases) a 
1-dimensional linear ranking supermartingale whose value in each location is 
equal to $c+d_{\loc}$, where $d_{\loc}$ is a location-specific constant. Since 
the expected change of $c$ in each loop step is $-0.5$, this is indeed a LRSM. 
Also, in both cases, a LEM of the form $x+d_{\loc}'$, again for some suitable 
location-specific constants $d_{\loc}'$, is a PVSM for the outer loop, as $x$ 
decreases and is non-negative within the outer loop and its expected change is 
non-negative within the inner loop (more precisely, the inner-loop expected 
change of $x$ zero in the left program and $-\frac{x}{4} - \frac{3}{4}$ in the 
right 
program). 
However, the variable $x$ is uniformly integrable within the inner loop of the 
right program while for the left program this does not hold: we show this 
in~\AppendixMaterial. The example shows that proving uniform integrability 
requires intricate reasoning about quantitative behaviour of the program. 
Moreover, as shown below, none of the two programs have a PVSM that typechecks.

\end{example}

\lstset{language=affprob}
\lstset{tabsize=2}
\newsavebox{\uniinto}
\begin{lrbox}{\uniinto}
\begin{lstlisting}[mathescape]
while $x\geq 0$ do
	$c:=1$;
	while $x\geq 1$ and $c\geq 1$ do
		if prob(0.5) then $x:=0$
			else $x:=2(x-1)$ fi;
		if prob(0.5) then $c:=0$ 
			else skip fi
	od;
	$x:=x-1$
od
\end{lstlisting}
\end{lrbox}
\newsavebox{\uniintt}
\begin{lrbox}{\uniintt}
	\begin{lstlisting}[mathescape]
while $x\geq 0$ do
	$c:=1$;
	while $x\geq 1$ and $c\geq 1$ do
		if prob(0.5) then $x:=0$
			else $x:=\frac{3}{2}(x-1)$ fi;
		if prob(0.5) then $c:=0$ 
			else skip fi
	od;
	$x:=x-1$
od
	\end{lstlisting}
\end{lrbox}
\begin{figure}[t]
	\usebox{\uniinto}
	\hspace{1cm}
	\usebox{\uniintt}
\caption{Examples of programs with (right) and without (left) uniformly 
integrable PVSMs.}
\label{fig:uniint}
\end{figure}

\lstset{language=affprob}
\lstset{tabsize=2}
\newsavebox{\unintthree}
\begin{lrbox}{\unintthree}
	\begin{lstlisting}[mathescape]
while $x\geq 0$ do
	while $y\geq 0$ do
		$z:=x$;
		while $z\geq 0$ do
			$z:=z-1$;
			$x:=x-1$
		od;
		$y:=y-1$
	od;
	$x:=x+$sample$(\mathit{Uniform[-3,1]})$
od
	\end{lstlisting}
\end{lrbox}
\newsavebox{\unintfour}
\begin{lrbox}{\unintfour}
	\begin{lstlisting}[mathescape]
while $x\geq 0 $ do
	$y:=x$;
	while $y\geq 1$ do
		$y:=y+$sample$(\mathit{Uniform[-3,1]})$
	od;
	$x:=x-1$
od
	\end{lstlisting}
\end{lrbox}
\begin{figure}[t]
	\begin{subfigure}{0.45\textwidth}
		\usebox{\unintfour}
		\vspace{0.2cm}
		\begin{tikzpicture}[x=1.8cm,y=1.4cm]
		\node[det] (init) at (0,0) {$\loc_0$};
		\node[det] (fin) at (0,-1) {$\loc^{\lout}$};
		\node[det] (as1) at (1,0) {$\loc_1$};
		\node[det] (inner) at (2,0) {$\loc_2$};
		\node[det] (as2) at (2,-1) {$\loc_3$};
		\node[det] (as3) at (3,0) {$\loc_4$};
		\draw[tran] (init) to node[font=\scriptsize,draw, fill=white, 
		rectangle,pos=0.5] {$x<0$} (fin);
		\draw[tran] (init) to node[font=\scriptsize,draw, fill=white, 
		rectangle,pos=0.5] {$x\geq 0$} (as1);
		\draw[tran] (as1) -- node[auto,font=\scriptsize] {$y:=x$} (inner);
		\draw[tran] (inner) to node[font=\scriptsize,draw, fill=white, 
		rectangle,pos=0.5] {$y<1$} (as3);
		\draw[tran] (inner) to node[font=\scriptsize,draw, fill=white, 
		rectangle,pos=0.5] {$y\geq 1$} (as2);
		\draw[tran,rounded corners] (as2) -- (1.25,-1) -- 
		node[font=\scriptsize, 
		label={[xshift=-0.4cm,yshift=-0.2cm,font=\scriptsize] 
		{$y:=\dots$}}] 
		{} 
		(inner);
		\draw[tran, rounded corners] (as3) -- (3,0.5) -- node[font=\scriptsize, 
		label={[font=\scriptsize, yshift=-0.1cm] {$x:=x-1$}}] {} (0,0.5) -- 
		(init);
		\end{tikzpicture}
		\caption{Example for slicing illustration: program and its pCFG.}
		\label{fig:uniint4}
	\end{subfigure}
	\hspace{0.5cm}
	\begin{subfigure}{0.45\textwidth}
	\usebox{\unintthree}
	\caption{Program where the outer loop does not have a PVSM that 
		typechecks.}
	\label{fig:uniint2}
	\end{subfigure}
\caption{Program illustrations.}
\label{fig:comp-ex}
\end{figure}

Indeed, taking a closer look at typesystem in~\cite{HolgerPOPL}, there are several reasons for typechecking of PVSM to fail. The major ones are:
\begin{compactenum}
\item
A PVSM $\lem$ for \PP{} $\program$ will not typecheck if $\program$ has a 
nested loop in which the value of $\lem$ can change unboundedly in a single 
step (see Figure~\ref{fig:uniint}). 
\item
A PVSM $\lem$ for  \PP{} $\program$ will not typecheck if $\program$ has a 
nested loop which itself has a nested loop in which some variable appearing in 
some expression in $\lem$ is modified, see 
Figure~\ref{fig:uniint2}.\footnote{Both these statements regarding typechecking 
failure are somewhat simplified, even in these two cases the PVSM might 
sometimes typecheck correctly, in case where the nested loops are followed by 
assignments which completely overwrite the effect of these loops, e.g. if 
the program in Figure~\ref{fig:uniint2} contained an assignment $x:=0$. 
However, the statements intuitively 
summarize the major reasons for typechecking failure. }
\end{compactenum}

Thus, the typechecking algorithm may rule out programs where the termination-controlling variable represents e.g. a length of an array, which can be doubled/halved in some sub-program do to (de)allocation, merging, or splitting.
To overcome the rather strict typechecking, we use the results on LexRSMs to define a new notion of compositional ranking supermartingales, which we call \emph{nob-negative compositional (NC)} supermartingales. For the sake of generality, we allow NC martingales to be general measurable maps, not necessarily propositionally linear.

\begin{definition}
\label{def:nonneg-comp}
A 1-dimensional measurable map $\lem$ is an NC supermartingale (NCSM) for a program $\program$ supported by an invariant map $\inv$ if there exists $\eps>0$ such that $\lem$ is:
\begin{compactenum}
\item  non-negative in each $(\loc,\vec{x})$ where $\loc$ is a location of $\pCFG_{\program}$ and $\vec{x}\in \inv(\loc)$;
\item 
 $\eps$-$\inv$-ranking in each location $\loc\in\slice{(\program)}$; and
\item 
  $\inv$-unaffecting in each $\loc\in\loops(\program)$.
\end{compactenum}
A (propositionally) linear NCSM (LinNCSM) is a NCSM which is also a (propositionally) linear expression map.
\end{definition}

We can prove that NCSMs are a sound method for proving a.s. termination in a compositional way, without any additional assumptions.

\begin{theorem}
\label{thm:nonneg-comp}
Let $\program$ be a \PP{} of the form $\textbf{while } \Psi \textbf{ do } 
\programbody \textbf{ od}$. Assume that each nested loop of $\program$ 
terminates almost surely from each reachable configuration, and that there 
exists a NCSM for $\program$ supported by some invariant map. Then $\program$ 
terminates almost surely from each initial configuration.
\end{theorem}
\begin{proof}[Proof (Key Idea)]
Let $\{X_i\}_{i=0}^{\infty}$ be a stochastic process returning the value of 
NCSM $\lem$ in step $i$. Then $\{X_i\}_{i=0}^{\infty}$ is a (non-strict) $1$-dimensional $\eps$-LexRSM for the termination time of the program, for some $\eps>0$. Since all sub-loops of $\program$ terminate, with probability one each run has level $<2$ in infinitely many steps. From theorem $\ref{thm:lexrsm-main}$ it follows that $\probm^\sigma_{\vecinit}(\ttime<\infty)=1$, for all $\sigma$ and $\vecinit$.
\end{proof}

Hence, NCSMs effectively trade the uniform integrability condition for 
non-negativity over the whole program. We believe that the latter condition is 
substantially easier to impose and check automatically, especially in the case 
of affine probabilistic programs and (propositionally) linear NCSMs: for 
\APP{}s, synthesizing NCSM entails synthesizing sufficient program invariants 
(for which there is a good automated tool support~\cite{FG10:aspic})  encoding 
the ranking, 
unaffection, and non-negativity conditions into a collection of linear 
constraints (as for general LinLexRSMs in Section~\ref{sec:algo}). 
Figures~\ref{fig:uniint} and~\ref{fig:uniint2} show instances where attempts to 
to prove of a.s. termination via PVSMs fail while proofs via LinNCSMs work. 

\begin{example}
For all programs in Figures~\ref{fig:uniint} and~\ref{fig:uniint2} there are 
LinNCSMs for all the loops in the program, which shows that the programs 
terminate 
a.s. In Figure~\ref{fig:uniint}, for both programs the inner loops have 
LinNCSMs of the form $c+d_{\loc}$, for $d_{\loc}$ a location-specific constant, 
while the outer loops have LinNCSMs of the form $x+d_{\loc}'$. In 
Figure~\ref{fig:uniint2} the program similarly has LinNCSMs defined, proceeding 
from the innermost loop and neglecting the location-specific constants, by 
variables $z$, $y$, $x$.
\end{example}

Using LinNCSMs we can devise the following simple algorithm for compositional proving of almost-sure termination of \APP{}s:

\begin{algorithm}
\SetKwInOut{Input}{input}\SetKwInOut{Output}{output}
\DontPrintSemicolon

\Input{An \APP{} $\program$ together with an invariant map $\inv$.}
$d\leftarrow \text{ depth of loop nesting in $\program$}$\;
\For{$i \leftarrow d$ \KwTo $0$}{
\ForEach{sub-loop $\program'$ of $\program$ nested $i$ levels below the main loop}{
\label{algoline:one}
$\linsystem \leftarrow$ system of lin. constraints encoding the existence of LinNCSM for $\program'$ supported by $\inv$\;
\label{algoline:two}\If{$\linsystem$ not solvable}{check existence of a PVSM for $\program'$~\cite{HolgerPOPL}
\;
\lIf{PVSM does not exist}{\Return{``cannot prove a.s. termination of $\program$''}}
}
}
}
\Return{``$\program$ terminates a.s.'' }
\caption{Compositional Termination Proving}
\label{algo:compositional}
\end{algorithm}

The soundness of the algorithm follows from Theorems~\ref{thm:nonneg-comp} and 
\ref{thm:holger-comp}. Note that we use the PVSM-based algorithm 
of~\cite{HolgerPOPL} as a back-up sub-procedure for the case when LinNCSM-based 
proof fails. Hence, Algorithm~\ref{algo:compositional} can compositionally 
prove a.s. termination of strictly larger class of programs than the PVSM-based 
algorithm alone.

To summarize, the novelty of NCSMs is the following:
\begin{compactenum}
\item
NCSMs allow compositional, fully automated proofs of a.s. termination without the need for reasoning about uniform integrability.
\item LinNCSMs are capable of proving a.s. termination of programs for which no uniformly integrable PVSMs exist (and hence the method of~\cite{HolgerPOPL} cannot be used at all on such programs).
\item LinNCSMs are capable of proving a.s. termination of programs for which the method of~\cite{HolgerPOPL} cannot be applied in an automated way, do to failure of the typechecking procedure.
\end{compactenum}

\subsection{Multidimensional Compositional Ranking}

 Above, we defined NCSMs as one-dimensional objects, to make them analogous to PVSMs for better comparison. However, we can also define a multi-dimensional version of NCSMs, to take advantage of the fact that LexRSMs can handle loops for which no 1-dimensional linear RSM exists. We say that, given an invariant map $\inv$, an $n$-dimensional measurable map is $\eps$-$\inv$-ranking in a location $\loc$ if for each $\vec{x}\in \inv(\loc)$ and each gen. transition $\tilde{\tau}$ outgoing from $\loc$ there exists $1\leq j \leq n$ such that $\tilde\tau$ is $\eps$-ranked by $\lem_j$ and for each $j'<j$ we have that  $\tilde{\tau}$ is unaffected by $\lem_{j'}$.

\begin{definition}
An $n$-dimensional measurable map $\vec{\lem}=(\lem_1,\dots,\lem_n)$ is an NC supermartingale (NCSM) for a program $\program$ supported by an invariant map $\inv$ if there exists $\eps>0$ such that:
\begin{compactenum}
	\item  for each $1\leq i \leq n$, $\lem_i$ is non-negative in each location of $\pCFG_{\program}$;
	\item 
	$\vec{\lem}$ is $\eps$-$\inv$-ranking in each location $\loc\in\slice{(\program)}$; and
	\item 
	for each $1\leq i \leq n$, $\lem_i$ is
	unaffected in each $\loc\in\loops(\program)$.
\end{compactenum}
An $n$-dimensional linear NCSM (LinNCSM) is an $n$-dimensional NCSM which is also a linear expression map.
\end{definition}

The following theorem can be proved in essentially the same way as Theorem~\ref{thm:nonneg-comp}.

\begin{theorem}
Let $\program$ be a \PP{} of the form $\textbf{while } \Psi \textbf{ do } 
\programbody \textbf{ od}$. Assume that each nested loop of $\program$ terminates almost surely from each reachable configuration, and that there exists a NCSM for $\program$ supported by some invariant map. Then $\program$ terminates almost surely.
\end{theorem}

For \APP{}s, we can generalize Algorithm~\ref{algo:compositional} by changing line~\ref{algoline:one} to ``check existence of a multi-dimensional LinNCSM for $\program'$'' and line~\ref{algoline:two} to ``\textbf{if} a multi-dimensional LinNCSM does not exist.'' The check of existence of a multi-dimensional LinNCSM for $\program'$ can be done by algorithm presented in Section~\ref{sec:algo}, modified so as to only pursue ranking for generalized transitions outgoing from locations belonging to $\slice(\program')$. (I.e., only these gen. transitions have $\eps_{\tau}$ included in the objective function, and algorithm terminates once all such transitions are ranked.) 




\section{Bounds on Expected Termination Time}

As shown in Example~\ref{ex:infinite-time}, already LinLexRSM maps are capable 
of proving almost-sure termination of programs whose expected termination time 
is infinite. However, it is often desirable to obtain bounds on expected 
runtime of a program. In this section, we present a LexRSM-based proof rule for 
obtaining bounds on expected runtime, and we show how to automatize the usage 
of this proof rule to obtain bounds on expected runtime in a 
subclass of \PP{}s.

As in the case of a.s. termination we start with general mathematical statement 
about LexRSMs. We define a restricted class of strict LexRSMs with \emph{bounded 
expected conditional increase} property. Recall from 
definition~\ref{def:lexrsm} that strict
LexRSM for a stopping time $\stime$ is characterized by the possibility to 
a.s. partition, for each $i\in \Nset_0$, the set $\{\omega\in \Omega\mid 
\stime(\omega)>i \}$ into $n$ sets $L^i_1,\dots,L^i_d$ such that, intuitively, 
on $L^i_j$ the conditional expectation of $\vec{X}_{i+1}[j]$ given $\genfilt_i$ 
is smaller than $\vec{X}_i[j]$, and for all $j'<j$, on $L^i_{j'}$ the 
conditional expectation of $\vec{X}_{i+1}[j]$ given $\genfilt_i$ 
is no larger than $\vec{X}_i[j]$. This leaves the opportunity of conditional 
expectation of $\vec{X}_{i+1}[j]$ being larger than $\vec{X}_i[j]$ on 
$L^i_{j''}$ with $j''>j$. The conditional expected increase property bounds the 
possibility of this increase.

\begin{definition}
	\label{def:lexrsm-eci}
Let $\{\vec{X}_{i}\}_{i=0}^{\infty}$ be an 
$n$-dimensional strict LexRSM for some stopping time $\stime$, defined w.r.t. some 
filtration $\{\genfilt_i\}_{i=0}^{\infty}$. We say that 
$\{\vec{X}_{i}\}_{i=0}^{\infty}$ has \emph{$\vec{c}$-bounded expected 
conditional 
increase (ECI),} for 
some non-negative vector $\vec{c}\in\Rset^d$, if there exists an instance $(\{\vec{X}_{i=0}^{\infty},\{L_1^i,\dots,L_{n+1}^i\}_{i=0}^{\infty})$ of the strict LexRSM (i.e. $L^i_{n+1}=\emptyset$ for all $i$) such that for each $i\in \Nset_0 $ and 
each $1\leq 
j \leq n $ it holds 
$\E[\vecseq{X}{i+1}{j}\mid \genfilt_i]\leq \vecseq{X}{i}{j}+\vec{c}[j]$ on 
$L^i_{j''}$, 
for all $j''>j$ 
(here $L^i_1,\dots,L^i_n$ are as in Definition~\ref{def:lexrsm}).
\end{definition}

For strict LexRSMs with $\vec{c}$-bounded ECI we have the following result. For 
simplicity, 
we formulate the result for 1-LexRSMs, though it is easy to prove analogous 
result for general $\eps$-LexRSMs, $\eps>0$, at the cost of obtaining less 
readable formula.

\begin{theorem}
\label{thm:runtime-bound}
Let $\{\vec{X}_{i}\}_{i=0}^{\infty}$ be an 
$n$-dimensional strict LexRSM with $c$-bounded ECI for some stopping time $\stime$. 
Then $\E[\stime]\leq  
\sum_{j=1}^{n}\E[\vecseq{X}{0}{j}]\cdot(\vec{c}[j]+1)^{n-j}$.
\end{theorem}
\begin{proof}
Fix an instance $(\{\vec{X}_{i=0}^{\infty},\{L_1^i,\dots,L_{n+1}^i\}_{i=0}^{\infty})$ satisfying Definition~\ref{def:lexrsm-eci}. Denote $\noofdecrank_j(\omega)$ the number of steps $i$ in which $\omega\in 
L_j^i$. Since by Theorem~\ref{thm:lexrsm-main} the existence of strict LexRSM entails 
$\probm(\stime<\infty)=1$, the value $\noofdecrank_j(\omega)$ is a.s. finite for all $1\leq j \leq n$. 
We prove that for each $1\leq j \leq n$ it holds $\E[\noofdecrank_j]\leq 
\vec{c}[j]\cdot\left(\sum_{j'<j}\E[\noofdecrank_{j'}]\right) + 
\E[\vecseq{X}{0}{j}].$ Since $\stime(\omega)=\sum_{1\leq j \leq n} 
\noofdecrank_j(\omega)$, for each $\omega\in \Omega$ (and hence, due to 
linearity of expectation $\E[\stime]=\sum_{1\leq j \leq n} 
E[\noofdecrank_j]$), the statement of the 
Theorem follows by an easy induction. 

To prove the required inequality, let $\nodecrank{k}{j}(\omega)$ be the number 
of steps $i$ \emph{within the first $k$} steps such that $\omega\in L_j^i$. We 
prove, by induction on $k$, that for each $k$ it holds 
$\E[\nodecrank{k}{j}]\leq 
\vec{c}[j]\cdot\left(\sum_{j'<j}\E[\nodecrank{k}{j'}] \right)+ 
\E[\vecseq{X}{0}{j}] - \E[\vecseq{X}{k}{j}]$. Once this is proved, the 
desired inequality follows by taking $k$ to $\infty$, since 
$\lim_{k\rightarrow \infty}\E[\nodecrank{k}{j}] = \E[\noofdecrank_j]$ and 
$\lim_{k\rightarrow \infty}\E[\vecseq{X}{k}{j}] \geq 0$.

The base case $k=0$ is simple as both sides of the inequality are zero. Assume 
that the inequality holds for some $k\geq 0$. We have 
$\E[\nodecrank{k+1}{j}]=\E[\nodecrank{k}{j}]+\probm{(L_j^{k})}$, so from 
induction hypothesis we get 
\begin{equation}
\label{eq:time1}
\E[\nodecrank{k+1}{j}]\leq 
\vec{c}[j]\cdot\left(\sum_{j'<j}\E[\nodecrank{k}{j'}] \right)+ 
\E[\vecseq{X}{0}{j}] - \E[\vecseq{X}{k}{j}] + \probm{(L_j^{k})}.\end{equation} 
Now denote 
$L^k_{<j} = L^k_1 \cup \dots\cup L^k_{j-1}$ and $L^k_{>j}= 
L^k_{j+1}\cup\dots\cup L^k_{d}$. We have  
$\E[\vecseq{X}{k}{j}] = \E[\vecseq{X}{k}{j}\cdot 
\indicator{ L_{<j}^k}] + \E[\vecseq{X}{k}{j}\cdot 
\indicator{ L_j^k}] + \E[\vecseq{X}{k}{j}\cdot 
	\indicator{L_{>j}^k}] \geq 
\E[\vecseq{X}{k+1}{j}\cdot 
	\indicator{ L_{<j}^k}] -\vec{c}[j]\cdot \probm(L_{<j}^k) + 
	\E[\vecseq{X}{k+1}{j}\cdot 
	\indicator{ L_{j}^k}] + \probm(L_{j}^k) + \E[\vecseq{X}{k+1}{j}\cdot 
	\indicator{ L_{>j}^k}]= \E[\vecseq{X}{k+1}{j}] -\vec{c}[j]\cdot 
	\probm(L_{<j}^k)+ \probm(L_{j}^k)$. Plugging this 
	into~\ref{eq:time1} yields
\begin{align*}
\E[\nodecrank{k+1}{j}]&\leq 
\vec{c}[j]\cdot\left(\sum_{j'<j}\E[\nodecrank{k}{j'}] \right)+ 
\E[\vecseq{X}{0}{j}] - \E[\vecseq{X}{k+1}{j}] + \vec{c}[j]\cdot 
\probm(L_{<j}^k)\\
&=\vec{c}[j]\cdot\left(\sum_{j'<j}\E[\nodecrank{k+1}{j'}] \right)+ 
\E[\vecseq{X}{0}{j}] - \E[\vecseq{X}{k+1}{j}].
\end{align*}
%
\end{proof}

To transfer this mathematical result to probabilistic programs, we want to 
impose a restriction on LexRSM maps that ensures that all components of a 
LexRSM map have, from each reachable configuration, an expected one-step 
increase of at most $c$. Here $c$ can be a constant, but it can also be a value 
that depends on the initial configurations: this is to handle cases where some 
variables are periodically reset to a value related to the initial variable 
values, such as variable $z$ in Figure~\ref{fig:uniint2}. To this end, let 
$\program$ be a \PP{} with a pCFG $\pCFG_\program$ and let 
$\vec{\lem}=(\lem_1,\dots,\lem_n)$ be an 
$n$-dimensional $1$-LexRSM map for $\program$. Consider an $n$-dimensional 
vector 
$\vec{\bar{c}}=(\bar{c}_1,\dots,\bar{c}_n)$ whose each component is an 
expression over variables of 
$\program$. We say 
that $\vec{\lem}$ has 
$\vec{\bar{c}}$-bounded ECI w.r.t. invariant map $\inv$ if the following holds 
for each initial configuration $(\locinit,\vecinit)$ with $\vecinit\in 
\vecinitset$: for 
each 
configuration $(\loc,\vec{x})$ with $\vec{x}\in \inv(\loc)$ and generalized 
transition $\tilde{\tau}$ 
of 
$\pCFG_\program$ outgoing from $(\loc,\vec{x})$ it holds that if $j$ is 
the smallest index such that 
$\tilde{\tau}$ is $1$-ranked by $\lem_j$ from $(\loc,\vec{x})$, then for all 
$j'>j$ the gen. transition $\tilde{\tau}$ is $f$-ranked by $\lem_{j'}$ from 
$(\loc,\vec{x})$, where $f=-c_{j'}(\vecinit)$. From 
Theorem~\ref{thm:runtime-bound} we have the following:

\begin{corollary}
\label{col:runtime-progs}
Let $\program$ be a probabilistic program. Assume that there exists an 
$n$-dimensional $\eps$-LexRSM map $\vec{\lem}=(\lem_1,\dots,\lem_d)$ for 
$\program$ supported 
by some 
invariant map $\inv$, scuh that $\vec{\lem}$ has $\vec{\bar{c}}$-bounded ECI 
(w.r.t. $\inv$) 
for some vector of expressions $\vec{\bar{c}}=(\bar{c}_1,\dots,\bar{c}_n)$. 
Then under each scheduler $\sigma$ and for each initial valuation of program 
variables $\vecinit\in\vecinitset$ it holds $\E^{\sigma}_{\vecinit}[\ttime]\leq 
\sum_{j=1}^{n}\lem_j(\locinit,\vecinit)\cdot(\bar{c}_j(\vecinit))^{n-j}$.
\end{corollary}

Now assume that we have synthesized a $1$-LexRSM map $\vec{\lem}$ and we want 
to check if there exists $\bar{\vec{c}}$ such that $\vec{\lem}$ has 
$\bar{\vec{c}}$-bounded ECI. In the linear setting (i.e. the program is an 
\APP{}, masp $\lem$ and $\inv$ are linear, and we seek $\vec{\bar{c}}$ which is 
a vector of affine expressions) we can encode the existence of $\vec{\bar{c}}$ 
into a system of linear inequalities in the similar way as the existence of a 
linear LexRSM maps was encoded in Section~\ref{sec:lex-programs}. That is, we 
set up a linear template with unknown coefficients for each component of 
$\vec{\bar{c}}$ and using Farkas's lemma we set up a system of linear 
constraints, which includes the unknown coefficients as variables, encoding the 
fact that $\vec{\lem}$ has $\bar{\vec{c}}$-bounded ECI. Details are provided 
in~\AppendixMaterial. In this way, we can check in polynomial time if 
Corollary~\ref{col:runtime-progs} can be applied to $\vec{\lem}$, and if yes, 
we can synthesize the witness vector $\bar{\vec{c}}$. Since $\bar{\vec{c}}$ 
consists of affine expressions, Corollary~\ref{col:runtime-progs} provides a 
polynomial (in the size of 
initial variable valuation) upper bound on expected runtime.



\section{Experimental Results}\label{sec:experiments}
We have implemented the algorithm of Section~\ref{sec:algo} (in C++) and present two sets of experimental results.
For all the experimental results we use the tool \textsc{Aspic}~\cite{FG10:aspic} for 
invariant generation, 
and our algorithm requires linear-programming solver for which we use 
\textsc{CPlex}~\cite{cplex}.
All our experimental results were obtained on the following platform: 
Ubuntu16.04, 7.7GB, Intel-Core i3-4130 CPU 3.40GHz QuadCore 64-bit.

\smallskip\noindent{\em Abstraction of real-world programs.}
We consider benchmarks from abstraction of real-world non-probabilistic programs 
in~\cite{ADFG10:lexicographic}.
Note that there are no abstraction tools available for probabilistic programs 
(to the best of our knowledge). 
Hence we consider the abstract programs obtained from real-world benchmarks as 
considered in the benchmark suite of~\cite{ADFG10:lexicographic}.
Given these non-probabilistic programs we obtain probabilistic programs in two ways:
(a)~{\em probabilistic loops} where the existing while loops are made probabilistic 
by executing the existing statements with probability~1/2, and with remaining 
probability executing skip statements; and
(b)~{\em probabilistic assignments} where the existing assignments are perturbed uniformly in range $[-1,1]$ 
(i.e., we consider additional variables whose value is, in each loop iteration, generated by 
probabilistic assignment uniformly in the range $[-1,1]$, and we add such variable 
to the RHS of an existing assignment).
We report our results on twenty five benchmarks in Table~\ref{tab:exp1} (we 
consider around fifty benchmark examples and the results on the remaining ones 
are presented in \AppendixMaterial.
The experimental results show that the time taken by our approach is always less 
than $1/10$-th of second.
In the table, along with the benchmark name, and time in seconds, we show 
whether a solution exists or not (i.e., whether linear lexicographic RSMs exist
or not), and if the solution exists we present the dimension of the lexicographic RSM
we obtain. 
The final two columns of the table represent whether the non-probabilistic 
program is extended with probabilistic loops and/or probabilistic assignments.

\smallskip\noindent{\em Synthetic examples of large programs.} 
The programs obtained as abstractions of real-world programs in the benchmarks 
as mentioned above have between 10-100 lines of code. 
To test the how does our approach scale with larger codes we consider synthetic examples 
of large probabilistic programs generated as follows.
Given $n$ additional Boolean-like variables, we consider probabilistic while loops,
with some nondeterministic conditional branches, and generate all possible $2^n$ if conditions
based on the Boolean variables. 
Hence given $n$ variables we have probabilistic programs of size $O(2^n)$.
For such programs we first run an invariant generation tool, followed by our algorithm.
In all these examples lexicographic RSMs exist, and has dimension at most~3.
Even for programs with around 12K lines of codes the total time taken is around 
one hour, where the invariant generation (i.e., running Aspic) takes the maximum time, 
and our algorithm requires around two minutes.
The results are presented in Table~\ref{tab:exp2} where we present the number of 
variables, then lines of code, followed by the time taken for invariant generation
by Aspic, then the time taken by our algorithm, and finally the total time.


\begin{center}
\begin{table}[]
  \centering
   \begin{tabular}{c|c|c|c|c|c}
    
{Benchmark} & {Time (s)} & {Solution} & {Dimension} & {Prob. loops} & {Prob. Assignments} \\\hline \hline
{alain} & {0.11} & {yes} & {2} & {yes} & {yes} \\\hline
{catmouse} & {0.08} & {yes} & {2} & {yes} & {yes} \\\hline
{counterex1a} & {0.1} & {no} & {} & {no} & {no} \\\hline
{counterex1c} & {0.11} & {yes} & {3} & {yes} & {yes} \\\hline
{easy1} & {0.09} & {yes} & {1} & {yes} & {yes} \\\hline
{exmini} & {0.09} & {yes} & {2} & {yes} & {yes} \\\hline
{insertsort} & {0.1} & {yes} & {3} & {yes} & {yes} \\\hline
{ndecr} & {0.09} & {yes} & {2} & {yes} & {yes} \\\hline
{perfect} & {0.11} & {yes} & {3} & {yes} & {yes} \\\hline
{\multirow{2}{*}{perfect2}} & {0.1} & {yes} & {3} & {yes} & {no} \\\cline{2-6}
{} & {0.11} & {no} & {} & {yes} & {yes} \\\hline
{real2} & {0.09} & {no} & {} & {no} & {no} \\\hline
{realbubble} & {0.22} & {yes} & {3} & {yes} & {yes} \\\hline
{realselect} & {0.11} & {yes} & {3} & {yes} & {yes} \\\hline
{realshellsort} & {0.09} & {no} & {} & {yes} & {no} \\\hline
{serpent} & {0.1} & {yes} & {1} & {yes} & {yes} \\\hline
{sipmabubble} & {0.1} & {yes} & {3} & {yes} & {yes} \\\hline
{speedDis2} & {0.09} & {no} & {} & {no} & {no} \\\hline
{speedNestedMultiple} & {0.1} & {yes} & {3} & {yes} & {yes} \\\hline
{speedpldi2} & {0.09} & {yes} & {2} & {yes} & {yes} \\\hline
{speedpldi4} & {0.09} & {yes} & {3} & {yes} & {yes} \\\hline
{speedSimpleMultipleDep} & {0.09} & {no} & {} & {no} & {no} \\\hline
{\multirow{2}{*}{speedSingleSingle2}} & {0.12} & {yes} & {2} & {yes} & {no} \\\cline{2-6}
{} & {0.1} & {no} & {} & {yes} & {yes} \\\hline
{\multirow{2}{*}{unperfect}} & {0.1} & {yes} & {2} & {yes} & {no} \\\cline{2-6}
{} & {0.16} & {no} & {} & {yes} & {yes} \\\hline
{wcet1} & {0.11} & {yes} & {2} & {yes} & {yes} \\\hline
{while2} & {0.1} & {yes} & {3} & {yes} & {yes} \\\hline
    
\end{tabular}
\caption{Experimental results for benchmarks from~\cite{ADFG10:lexicographic} extended with probabilistic loops and/or probabilistic assignments.} \label{tab:exp1}
\end{table}
\end{center}

\begin{center}
\begin{table}[]
  \centering
   \begin{tabular}{c|c|c|c|c}

{Variables} & {LOC} & {Inv Time (s)} & {Our Time (s)} & {Total Time (s)} \\\hline \hline
{2} & {20} & {0.06} & {0.03} & {0.08} \\\hline
{3} & {32} & {0.07} & {0.03} & {0.09} \\\hline
{4} & {56} & {0.08} & {0.04} & {0.11} \\\hline
{5} & {104} & {0.14} & {0.06} & {0.19} \\\hline
{6} & {200} & {0.36} & {0.1} & {0.46} \\\hline
{7} & {392} & {1.31} & {0.3} & {1.61} \\\hline
{8} & {776} & {7.56} & {0.7} & {8.25} \\\hline
{9} & {1544} & {33.07} & {2.5} & {35.57} \\\hline
{10} & {3080} & {164.09} & {8.77} & {172.86} \\\hline
{11} & {6152} & {817.92} & {35.37} & {853.29} \\\hline
{12} & {12296} & {4260.96} & {145.18} & {4406.14} \\\hline
\end{tabular}
\caption{Experimental results for synthetic examples.}\label{tab:exp2}
\end{table}
\end{center}

\section{Related Work}
In this section we discuss the related works.

\smallskip\noindent{\em Probabilistic programs and termination.}
In early works the termination for concurrent probabilistic programs was studied 
as fairness~\cite{SPH84}, which ignored precise probabilities.
For countable state space a sound and complete characterization of almost-sure termination 
was presented in~\cite{HS85}, but nondeterminism was absent.
A sound and complete method for proving termination of finite-state programs
was given in~\cite{EGK12}.
For probablistic programs with countable state space and without 
nondeterminism, the {\em Lyapunov ranking functions} provide a sound and 
complete method to prove positive termination~\cite{BG05,Foster53}.
For probabilistic programs with nondeterminism, but restricted to discrete probabilistic
choices, the termination problem was studied 
in~\cite{MM04,MM05}.
The RSM-based (ranking supermatingale-based) approach extending ranking functions was first 
presented in~\cite{SriramCAV} for probabilistic programs without non-determinism,
but with real-valued variables, and its extension for probabilistic programs
with non-determinism has been studied 
in~\cite{HolgerPOPL,CF17,CFNH16:prob-termination,CFG16,CNZ17,MM16:proofrule-arxiv}. 
Supermartingales were also 
considered for other liveness and safety 
properties~\cite{CVS16:martingale-recurrence-persistence,BEFH16:doob}.
While all these results deeply clarify the role of RSMs for probabilistic programs, 
the notion of lexicographic RSMs to obtain a practical approach for termination 
analysis for probabilistic programs has not been studied before, which we consider in 
this work.

\smallskip\noindent{\em Compositional a.s. termination proving.} 
A 
compositional rule for proving almost-sure termination was studied 
in~\cite{HolgerPOPL} under the uniform integrability assumption. 
In~\cite{MM05}, a soundness of the probabilistic variant rule is proved for 
programs with finitely many configurations.

\smallskip\noindent{\em Other approaches.}
Besides RSMs, other approaches has also been considered for probabilistic programs.
Logical calculi for reasoning about properties of 
probabilistic programs (including termination) were studied 
in~\cite{Kozen:prob-semantics,FH:prdl,Kozen:probabilistic-PDL,Feldman:propositional-probdl}
 and extended to programs with demonic non-determinism 
 in~\cite{MM04,MM05,KKMO16:wp-expected-runtime,OKKM16:recursive-prob-wp-calculus,GKI14:prob-semantics,
  DBLP:conf/sas/KatoenMMM10}. However, none of these approaches is readily 
  automatizable.
A sound approach~\cite{DBLP:conf/sas/Monniaux01} for almost-sure termination 
is to explore the exponential decrease of probabilities upon 
bounded-termination 
through abstract interpretation~\cite{DBLP:conf/popl/CousotC77}. A method for 
a.s. termination of weakly finite programs (where number of reachable 
configurations is finite from each initial configuration) based on 
\emph{patterns} was presented in~\cite{EGK12}.

\smallskip\noindent{\em Non-probabilistic programs.}
Termination analysis of non-probabilistic programs has also been extensively 
studied~\cite{PR04:transition-invariants, 
CPR06:terminator,DBLP:conf/cav/BradleyMS05,DBLP:conf/tacas/ColonS01, 
DBLP:conf/vmcai/PodelskiR04,DBLP:conf/pods/SohnG91,BMS05b,LJB01,KSTW10:compositional-transition-invariants,
 CPR11:termination-cacm}.
Ranking functions are at the heart of the termination analysis, and lexicographic 
ranking function has emerged as one of the most efficient and practical approaches
for termination analysis~\cite{CSZ13,ADFG10:lexicographic,GMR15:rank-extremal}, 
being used e.g. in 
the prominent \texttt{T2} temporal prover~\cite{BCIKP16:T2}.
In this work we extend lexicographic ranking functions to probabilistic 
programs,
and present lexicographic RSMs for almost-sure termination analysis of probabilistic programs
with non-determinism. Theoretical complexity of synthesizing lexicographic 
ranking functions in non-probabilistic programs was studied 
in~\cite{BG13:integer-ranking,BG15:lexicographic-complexity}.

\section{Conclusion and Future Work}
In this work we considered lexicographic RSMs for termination analysis of probabilistic
programs with non-determinism.
We showed it presents a sound approach for almost-sure termination, that is
algorithmically efficient, enables compositional reasoning about termination, 
and leads to approach that can handle realistic programs.
There are several interesting directions of future work.
Lexicographic ranking functions has been considered in several works to 
provide different practical methods for analysis of non-probabilistic programs.
First, while our work presents the foundations of lexicographic RSMs for probabilistic 
programs, extending other practical methods based on lexicographic ranking 
functions to 
lexicographic RSMs is an interesting direction of future work.
Second, while our algorithmic approaches focus on the linear case, it would be interesting
to consider other non-linear, and polynomial functions in the future.

\bibliographystyle{ACM-Reference-Format}
\bibliography{bibliography-master,new-popl18,new}
\clearpage
\appendix

\begin{center}
{\Large Supplementary Material}
\end{center}


\section{Details of Program Syntax}
In this subsection we present the details of the syntax of (affine) 
probabilistic 
programs.

Recall that $\vars$ 
is a collection of 
\emph{variables}.
Moreover, let $\mathcal{D}$ be a set of \emph{probability distributions} on 
real numbers.
The abstract syntax of affine
probabilistic programs (\APP s)
is given by the grammar in Figure~\ref{fig:syntax}, where
the expressions $\langle \mathit{pvar}\rangle$ and $\langle
\mathit{dist}\rangle$  range over $\vars$ and $\mathcal{D}$, respectively.
We allow for non-deterministic assignments, expressed by a statement $x:=
\text{\textbf{ndet($\mathit{dom}$)}}$, where $\mathit{dom}$ is a \emph{domain
specifier} determining the set from which the value can be chosen: for general 
programs it can be any Borel-measureable set, for \APP{}s it has to be an 
interval (possibly of infinite length).  
The grammar is such that $\langle \mathit{expr} \rangle$ 
may evaluate to an arbitrary affine expression over the
program variables.
Next, $\langle
\mathit{bexpr}\rangle$ may evaluate to an arbitrary propositionally linear
predicate. 

For general (not necessary affine) \PP{}s we set $\langle\mathit{expr}\rangle$ 
to be the set of all expressions permitted by the set of mathematical 
operations of the underlying language. Similarly, $\langle\mathit{bexpr} 
\rangle$ is the set of all predicates, as defined in Section~\ref{sec:prelim}.

The guard of each if-then-else statement is either $\star$, 
representing a (demonic) non-deterministic choice between the branches,
a keyword \textbf{prob}($p$), where $p\in [0,1]$ is a number given in decimal
representation (represents a probabilistic choice, where  the if-branch is
executed with probability $p$ and the then-branch with probability $1-p$), or
the guard is a propositionally linear predicate, in which case the statement
represents a standard deterministic conditional branching.

Regarding distributions, for each $d\in \mathcal{D}$ we assume the
existence of a program primitive denoted by '\textbf{sample($d$)}' implementing 
sampling from $d$. In practice, the distributions appearing in a program would 
be those for which sampling is 
provided by suitable libraries (such as 
uniform distribution over some interval, Bernoulli, geometric, etc.), 
but we 
abstract away from these implementation details. For the purpose of our 
analysis, it is sufficient that for each distribution $d$ appearing in the 
program the following characteristics: expected value $\expv[d]$ of $d$ and a 
set $SP_d$ containing \emph{support} of $d$  (the support of $d$ is the 
smallest 
closed set of real numbers whose complement has probability zero 
under $d$). \footnote{In 
particular, 
a support of a \emph{discrete} probability 
distribution $d$ is simply the at most countable set of all points on a real 
line that have positive probability under $d$. For continuous distributions, 
e.g. a normal distribution, uniform, etc., the support is typically either 
$\Rset$ or some closed real interval. } For \APP{}s, $SP_d$ is required to be 
an 
interval. 
\begin{figure}
\begin{align*}
\langle \mathit{stmt}\rangle &::= 
\langle \mathit{assgn} \rangle \mid \text{'\textbf{skip}'} \mid 
\langle\mathit{stmt}\rangle \, \text{';'} \, \langle \mathit{stmt}\rangle \\
&\mid   \text{'\textbf{if}'} \,
\langle\mathit{ndbexpr}\rangle\,\text{'\textbf{then}'} \, \langle
\mathit{stmt}\rangle \, \text{'\textbf{else}'} \, \langle \mathit{stmt}\rangle
\,\text{'\textbf{fi}'}
\\
&\mid  \text{'\textbf{while}'}\, \langle\mathit{bexpr}\rangle \,
\text{'\textbf{do}'} \, \langle \mathit{stmt}\rangle \, \text{'\textbf{od}'}
\\
\langle \mathit{assgn} \rangle &::= 
\,\langle\mathit{pvar}\rangle
\,\text{'$:=$'}\, \langle\mathit{expr} \rangle \mid 
\langle\mathit{pvar}\rangle \,\text{'$:=$}\,
\text{\textbf{ndet($\langle\mathit{dom}\rangle$)}'}
\\
&\mid \langle\mathit{pvar}\rangle \,\text{'$:=$}\,
\text{\textbf{sample($\langle\mathit{dist}\rangle$)}'}
%
%
%
%
%
%
\\
\vspace{0.5\baselineskip}
\langle\mathit{expr} \rangle &::= \langle \mathit{constant} \rangle \mid
\langle\mathit{pvar}\rangle
\mid \langle \mathit{constant} \rangle \,\text{'$\cdot$'} \,
\langle\mathit{pvar}\rangle
\\
&\mid \langle\mathit{expr} \rangle\, \text{'$+$'} \,\langle\mathit{expr} \rangle
\mid \langle\mathit{expr} \rangle\, \text{'$-$'} \,\langle\mathit{expr} \rangle
\\
%
%
\langle \mathit{bexpr}\rangle &::=  \langle \mathit{affexpr} \rangle \mid
\langle \mathit{affexpr} \rangle \, \text{'\textbf{or}'} \,
\langle\mathit{bexpr}\rangle
\vspace{0.5\baselineskip}
\\
%
\langle\mathit{affexpr} \rangle &::=  \langle\mathit{literal} \rangle\mid
\langle\mathit{literal} \rangle\, \text{'\textbf{and}'}
\,\langle\mathit{affexpr} \rangle
\\
\langle\mathit{literal} \rangle &::= \langle\mathit{expr} \rangle\,
\text{'$\leq$'} \,\langle\mathit{expr} \rangle \mid \langle\mathit{expr}
\rangle\, \text{'$\geq$'} \,\langle\mathit{expr} \rangle
\\
&\mid \neg \langle \mathit{literal} \rangle
\\
%
\langle\mathit{ndbexpr} \rangle &::= {\star}\mid
\text{'\textbf{prob($p$)}'} \mid \langle\mathit{bexpr} \rangle
\end{align*}
\caption{Syntax of affine probabilistic programs (\APP 's).}
\label{fig:syntax}
\end{figure}

\section{Details of Program Semantics}

\begin{remark}[Use of random variables]
In the paper we sometimes work with random variables 
that are functions of the type $R\colon\Omega \rightarrow S$ for some finite 
set $S$. These can be captured by the definition given in 
Section~\ref{sec:prelim} by identifying the 
elements of $S$ with distinct real numbers.\footnote{This is equivalent to 
saying that a function $R\colon \Omega\rightarrow S$, with $S$ finite, is a 
random variable if for each $s\in S$ the set $\{\omega\in \Omega\mid 
R(\omega)=s\}$ belongs to $\mathcal{F}$.} The exact choice of numbers is 
irrelevant in such a case, as we are not interested in, e.g. computing expected 
values of such random variables, or similar operations. 
\end{remark}

\paragraph*{From Programs to pCFGs}
To every probabilistic program $P$ we can assign a pCFG $\pCFG_P$ whose 
locations correspond to the values of the
program counter of $P$ and whose transition relation captures the behaviour of
$P$. We illustrate the construction for \APP{}s, for general programs it is 
similar. To obtain $\pCFG_{P}$, we first rename 
the variables in $P$ to 
$x_1,\dots,x_n$, where $n$ is the number of distinct variables in the program. 
The
construction of $\pCFG_P$ can be described inductively.
For each program $P$ the pCFG $\pCFG_P$ contains two distinguished
locations, $\ell^{\lin}_{P}$ and $\ell^{\lout}_{P}$, the latter one being always
deterministic, that intuitively represent the state of the program counter
before and after executing $P$, respectively. In the following, we denote by 
$\id_1$ a function such that for each $\vec{x}$ we have 
$\id_{1}(\vec{x})=\vec{x}[1]$.
\begin{compactenum}
\item {\em Deterministic Assignments and Skips.}
For $P= {x_j}{:=}{E}$ where $x_j$ is a program variable and $E$ is an 
expression, or $P = \textbf{skip}$, the pCFG $\pCFG_P$ consists only of
locations $\ell^{\lin}_P$ and $\ell^{\lout}_P$ (first assignment location, 
second one deterministic) and a
transition $(\ell^{\lin}_{P},\ell^{\lout}_P)$. In the first case, 
$\updates(\ell^{\lin}_{P},\ell^{\lout}_P)=(j,E)$.
\item {\em Probabilistic and Non-Deterministic Assignemnts}
For $P= {x_j}{:=}{\textbf{sample($d$)}}$ where $x_j$ is a program variable and 
$d$ is a distribution, the pCFG $\pCFG_P$ consists locations $\ell^{\lin}_P$ 
and $\ell^{\lout}_P$ and a
transition $\tau=(\ell^{\lin}_{P},\ell^{\lout}_P)$ with $\updates(\tau)=(j,d)$. For 
$P= 
{x_j}{:=}{\textbf{ndet($\mathit{dom}$)}}$, the construction is similar, with 
the only transition being $\tau=(\ell^{\lin}_{P}\ell^{\lout}_P)$ and 
$\updates(\tau)=(j,D)$, where 
$D$ is 
the set specified by the domain specifier $\mathit{dom}$.

\item {\em Sequential Statements.}
For $P = Q_1;Q_2$ we take the pCFGs $\pCFG_{Q_1}$, $\pCFG_{Q_2}$ and
join them by identifying the location $\ell^{\lout}_{Q_1}$ with
$\ell^{\lin}_{Q_2}$, putting $\ell^{\lin}_{P}=\ell^{\lin}_{Q_1}$ and
$\ell^{\lout}_{P}=\ell^{\lout}_{Q_2}$.

\item {\em While Statements.}
For $P = \textbf{while $\phi$ do }Q \textbf{ od}$ we add a new deterministic
location $\ell^{\lin}_{P}$ which we identify with $\ell^{\lout}_{Q}$, a new
deterministic location $\ell^{\lout}_{P}$, and transitions
$\tau=(\ell^{\lin}_{P},\ell^{\lin}_{Q})$,
$\tau'=(\ell^{\lin}_{P},\ell^{\lout}_{P})$ such that $G(\tau)=\phi$ and
$G(\tau')=\neg\phi$.

\item {\em If Statements.}
Finally, for $P = \textbf{if $\mathit{ndb}$ then }Q_1 \textbf{ else } Q_2
\textbf{ fi}$ we add a new location $\ell^{\lin}_{P}$ (which is not an 
assignment location) together with two
transitions $\tau_1 = (\ell^{\lin}_{P},\ell^{\lin}_{Q_1})$, $\tau_2 =
(\ell^{\lin}_{P},\ell^{\lin}_{Q_2})$, and we identify the locations 
$\ell^{\lout}_{Q_1}$ and $\ell^{\lout}_{Q_1}$ with $\ell^{\lout}_{P}$. (If both
$Q_j$'s consist of a single statement, we also identify $\ell^\lin_{P}$ with 
$\ell^{\lin}_{Q_j}$'s.) In this
case the newly added location $\ell^\lin_{P}$ is non-deterministic branching if 
and only 
if
$ndb$ is the keyword '$\star$'. If
$\mathit{ndb}$ is of the form $\textbf{prob($p$)}$, the location $\ell^\lin_{P}$
is probabilistic branching with $\probdist_{\ell^\lin_{P}}(\tau_1)=p$ and
$\probdist_{\ell^\lin_{P}}(\tau_2)=1-p$. Otherwise (i.e. if $\mathit{ndb}$ is a
predicate), $\ell^\lin_{P}$ is a deterministic location
with $G(\tau_1)=\mathit{ndb}$ and $G(\tau_2)=\neg \mathit{ndb}$.
\end{compactenum}
Once the pCFG $\pCFG_P$ is constructed using the above rules, we put
$G(\tau)=\textit{true}$ for all transitions $\tau$ outgoing from deterministic
locations whose guard was not set in the process, and finally we add a self-loop
on the location $\ell^{\lout}_P$. This ensures that the assumptions in
Definition~\ref{def:stochgame} are satisfied.
Furthermore note that for pCFG obtained for a program $P$, since the only
branching is if-then-else branching, every location $\loc$ has at most two
successors $\loc_1,\loc_2$.

\def\xx{\ref{prop:conditional-exp-existence}}

\section{Proof of Proposition~\xx}

We first recall the general statement of the Radon-Nikodym theorem. Given two measurable spaces\footnote{A generalization of a probability space where the measure of $\Omega$ does not have to be 1, but any non-negative number or even infinity.} $(\Omega,\genfilt,\mu)$ and $(\Omega,\genfilt,\nu)$, we say that $\nu$ dominates $\mu$, written $\mu<<\nu$ if for all $A\in \genfilt$, $\nu(A)=0$ implies $\mu(A)=0$. Radon-Nikodym theorem states that if both $\mu$ and $\nu$ are sigma-finite (that is, $\Omega$ is a union of countably many sets of finite measure under $\nu$ and $\mu$), then $\mu<<\nu$ implies that there exists an almost-surely unique $\genfilt$-measurable function $f\colon \Omega\rightarrow [0,\infty)$ such that for each $A \in\genfilt$, the Lebesgue integral of the function $f\cdot\indicator{A}$ in measurable space  $(\Omega,\genfilt,\nu)$ is equal to $\mu(A)$. The function $f$ is called a Radon-Nikodym derivative of $\mu$ w.r.t. $\nu$, and we denote in by $\frac{d\mu}{d\nu}$.

Now assume that $X$ is a non-negative real-valued random variable in some probability space $(\Omega,\genfilt,\probm)$ and $\genfilt'$ is a sub-sigma algebra of $\genfilt$. Note that $\probm$ is sigma-finite. Define a measure $\mu$ on $\genfilt'$ by putting $\mu(A)=\E[X\cdot\indicator{A}]$, for each $A\in \genfilt$ (here $\E$ is the expectation operator, i.e. the Lebesgue integral, in probability space $(\Omega,\genfilt',\probm')$, where $\probm'$ is a restriction of $\probm$ to $\genfilt'$). Then $\mu$ is sigma-finite: indeed, for any $n\in \Nset$ let $A_n = \{\omega\in\Omega\mid X(\omega)\leq n\}$. Then $\mu(A_n)\in [0,n]$, in particular it is finite, and since $X$ is real-valued, we have $\Omega=\bigcup_{n=1}^{\infty} A_n$. Hence, $\frac{d\mu}{d\probm'}$ exists and is almost-surely unique. It is now easy to check that $\frac{d\mu}{d\probm'}$ satisfies the condition defining the conditional expectation $\E[X\mid \genfilt']$: indeed, the condition is equivalent to $\E[X\mid \genfilt']$ being a derivative of $\mu$ w.r.t. $\probm'$. This concludes the proof.


\section{Computations for the proof of Theorem~\ref{THM:LEXRSM-MAIN}}

Recall that we aim to prove equation~\eqref{eq:lexrsm-soundness-main}.

For 
$k=\fixn{i}$ the sum on the right-hand side equals $0$, so the 
inequality immediately follows from the definition of $Y_k$. Now assume 
that~\eqref{eq:lexrsm-soundness-main} holds for some $k\geq \fixn{i}$. We have that 
\begin{align}
\label{eq:lexrsm-ind-1}
\E[Y_{k+1}] &= \underbrace{\E[Y_{k+1}\cdot \indicator{\Omega\setminus D}]}_{=0=\E[Y_k\cdot\indicator{\Omega\setminus D}]} +  \underbrace{\E[Y_{k+1}\cdot \indicator{D \cap \{F \leq k\}}]}_{=\E[Y_k\cdot\indicator{D\cap \{F\leq k\}}]} + \underbrace{\E[Y_{k+1}\cdot \indicator{D \cap \{F > k\}}]}_{=\E[\vecseq{X}{k+1}{\fixn{j}}\cdot\indicator{D\cap \{F>k \} }]},
\end{align}
where the equality $\E[Y_{k+1}\cdot \indicator{D \cap \{F \leq k\}}] =\E[Y_k\cdot\indicator{D\cap \{F\leq k\}}] $ follows from the fact that $Y_{k+1}(\omega)=Y_k({\omega})=\vecseq{X}{F(\omega)}{\fixn{j}}(\omega)$ for $\omega\in \{F\leq k\}$, and similarly for the last term. We prove that 
\begin{equation}
\label{eq:lexrsm-ind-2}
\E[\vecseq{X}{k+1}{\fixn{j}}\cdot \indicator{D\cap \{F > k \}}] \leq \E[Y_k\cdot \indicator{D\cap \{F > k \}} -\eps\cdot \indicator{D\cap \{F>k\} \cap \{\levelrank{}{k}= \fixn{j}\} }] .
\end{equation}
Indeed, it holds
\begin{equation}
\label{eq:lexrsm-ind-3}
\E[\vecseq{X}{k+1}{\fixn{j}}\cdot \indicator{D\cap \{F > k \}}] = \E[\vecseq{X}{k+1}{\fixn{j}}\cdot \indicator{D\cap \{F > k \} \cap \{\levelrank{}{k}=\fixn{j} \} }]  + \E[\vecseq{X}{k+1}{\fixn{j}}\cdot \indicator{D\cap \{F > k \} \cap \{\levelrank{}{k}>\fixn{j}\} }], 
\end{equation}
since $\levelrank{\omega}{k}\geq \fixn{j}$ for all $\omega \in \{F>k\}$. Since the set $D\cap \{F > k \} \cap \{\levelrank{}{k}=\fixn{j} \}$ is $\genfilt_k$-measurable, we get
\begin{align}
\E[\vecseq{X}{k+1}{\fixn{j}}\cdot \indicator{D\cap \{F > k \} \cap \{\levelrank{}{k}=\fixn{j} \} }] &= \E[\E[\vecseq{X}{k+1}{\fixn{j}}\mid \genfilt_k]\cdot \indicator{D\cap \{F > k \} \cap \{\levelrank{}{k}=\fixn{j} \} }] \label{eq:lexrsm-ind-4}\\
&\leq \E[(\vecseq{X}{k}{\fixn{j}} - \eps)\cdot \indicator{D\cap \{F > k \} \cap \{\levelrank{}{k}=\fixn{j} \}}] \label{eq:lexrsm-ind-5}\\
&=\E[(Y_k - \eps)\cdot \indicator{D\cap \{F > k \} \cap \{\levelrank{}{k}=\fixn{j} \}}] \label{eq:lexrsm-ind-6},
\end{align}
where~\eqref{eq:lexrsm-ind-4} follows from the definition of conditional expectation~\eqref{eq:cond-exp},~\eqref{eq:lexrsm-ind-5} follows from the definition of $\{\levelrank{}{k}=\fixn{j}\}$, and~\eqref{eq:lexrsm-ind-6} holds since $Y_k(\omega)=\vecseq{X}{k}{\fixn{j}}(\omega)$ for $\omega$ with $F(\omega)> k$. Almost identical argument shows that
\begin{equation}
\label{eq:lexrsm-ind-7}
\E[\vecseq{X}{k+1}{\fixn{j}}\cdot \indicator{D\cap \{F > k \} \cap \{\levelrank{}{k}>\fixn{j} \} }] \leq \E[Y_k\cdot \indicator{D\cap \{F > k \} \cap \{\levelrank{}{k}>\fixn{j} \}}].
\end{equation}
Plugging~\eqref{eq:lexrsm-ind-6} and~\eqref{eq:lexrsm-ind-7} into~\eqref{eq:lexrsm-ind-3} yields~\eqref{eq:lexrsm-ind-2}. Now we can plug~\eqref{eq:lexrsm-ind-2} into~\eqref{eq:lexrsm-ind-1} to get
\begin{align}
\E[Y_{k+1}]&\leq \E[Y_k] - \eps\cdot \E[\indicator{D \cap \{F>k\} \cap \{\levelrank{}{k} = \fixn{j} \} }] = \E[Y_k] - \eps\cdot \probm(D \cap \{F>k\} \cap \{\levelrank{}{k} = \fixn{j} \} ) \nonumber\\
&\leq \fixn{B}\cdot \probm(D) - \eps\cdot\left(\sum_{\ell=0}^{k-\fixn{i}} \ell\cdot\probm(D 
\cap \{F\geq k\} \cap \{\noofdec_k = \ell\})\right) -\eps\cdot \probm(D \cap \{F>k\} \cap \{\levelrank{}{k} = \fixn{j} \}),
\end{align}
where the last inequality follows from induction hypothesis. Hence, using $D_{k,\ell}$ as a shorthand for $D 
\cap \{F\geq k\} \cap \{\noofdec_k = \ell\}$, to prove~\eqref{eq:lexrsm-soundness-main} it remains to show that
\begin{equation}
\label{eq:lexrsm-ind-8}
\sum_{\ell=0}^{k-\fixn{i}} \ell\cdot\probm(D_{k,\ell}) + \probm(D \cap \{F>k\} \cap \{\levelrank{}{k} = \fixn{j} \}) = \sum_{\ell=0}^{k+1-\fixn{i}} \ell\cdot\probm(D_{k+1,\ell}).
\end{equation}
The left-hand side of~\eqref{eq:lexrsm-ind-8} is equal to
\begin{align}
&\phantom{+}\;\sum_{\ell=0}^{k-\fixn{i}} \ell\cdot\probm(D_{k,\ell} \cap \{\levelrank{}{k}=\fixn{j}\}) + \sum_{\ell=0}^{k-\fixn{i}} \ell\cdot\probm(D_{k,\ell}\cap \{\levelrank{}{k}>\fixn{j}\} )\nonumber \\ 
&+\sum_{\ell=0}^{k-\fixn{i}}\probm(\underbrace{D \cap \{F>k\} \cap \{\levelrank{}{k} = \fixn{j}\} \cap \{\noofdec_k = \ell\}}_{=D_{k,\ell} \cap \{\levelrank{}{k} = \fixn{j}\}})
\nonumber \\
&= \sum_{\ell=0}^{k-\fixn{i}}(\ell+1) \cdot\probm(D_{k,\ell}\cap \{\levelrank{}{k} = \fixn{j}\}) +  \sum_{\ell=0}^{k-\fixn{i}} \ell\cdot\probm(D_{k,\ell}\cap \{\levelrank{}{k}>\fixn{j}\} ) \nonumber\\
&=\sum_{\ell=0}^{k-\fixn{i}} (\ell+1)\cdot \probm{(D_{k+1,\ell+1} \cap \{\levelrank{}{k} = \fixn{j}\})} +  \sum_{\ell=0}^{k-\fixn{i}} \ell\cdot\probm(D_{k+1,\ell}\cap \{\levelrank{}{k}>\fixn{j}\} ) \label{eq:lexrsm-ind-9}\\
&= \sum_{\ell=1}^{k+1-\fixn{i}} \ell\cdot \probm{(D_{k+1,\ell} \cap \{\levelrank{}{k} = \fixn{j}\})} +  \sum_{\ell=0}^{k-\fixn{i}} \ell\cdot\probm(D_{k+1,\ell}\cap \{\levelrank{}{k}>\fixn{j}\} )
\nonumber \\
&= (k+1-\fixn{i})\cdot \probm{(D_{k+1,k+1-\fixn{i}}\cap\{\levelrank{}{k} =\fixn{j} \} )} + \sum_{\ell=1}^{k-\fixn{i}} \ell \cdot \probm(D_{k+1,\ell})  \nonumber\\
&= (k+1-\fixn{i})\cdot\probm(D_{k+1,k+1-\fixn{i}}) + \sum_{\ell=1}^{k-\fixn{i}}\ell\cdot\probm(D_{k+1,\ell}) = \text{ right-hand side of~\eqref{eq:lexrsm-ind-8}} \label{eq:lexrsm-ind-11}.
\end{align}
The individual steps in the above computation are justified as follows: 
in~\eqref{eq:lexrsm-ind-9} we use the facts that for all $\omega$'s whose level 
in step $k$ is $\fixn{j}$ it holds that 
$\noofdec_k(\omega)+1=\noofdec_{k+1}(\omega)$, and similarly, for $\omega$'s 
whose level in step $k$ is $>\fixn{j}$ it holds 
$\noofdec_k(\omega)=\noofdec_{k+1}(\omega)$. Moreover, for all $\omega\in 
D_{k,\ell}$ it holds that $\levelrank{k}{\omega}\geq \fixn{j} \Rightarrow 
F(\omega)\geq k+1$. Finally, in~\eqref{eq:lexrsm-ind-11} we use the fact that 
all $\omega\in D_{k+1,k+1-\fixn{i}}$ need to have level $\fixn{j}$ in step $k$, 
since otherwise such an $\omega$ would need to have level $\fixn{j}$ for at 
least $k+1-\fixn{i}$ times within steps $\{\fixn{i},\fixn{i}+1,\dots,k-1\}$, 
but there are $k-\fixn{i}$ such steps, a contradiction. This concludes the 
proof of~\eqref{eq:lexrsm-soundness-main}.

\section{Complexity Clarification for Theorem~\ref{THM:ALGO}}

Since instances of linear programming is in~\textsc{PTIME}, it remains to show 
that each system $\linsystem_{\tilde{\tau}}$ is constructible in polynomial 
time. In~\cite{CFNH16:prob-termination} it shown that this can be done provided 
that guard of each transition in pCFG is a propositionally linear predicate. 
Now all transition guards in $\pCFG_\program$ are of the form $\phi$ or 
$\neg\phi$, where $\phi$ is a guard of a conditional or of while-loop in 
$\program$. If $\phi$ is a linear assertion, then $\neg\phi$ can be converted 
into a propositionally linear predicate in polynomial time, after which the 
construction of~\cite{CFNH16:prob-termination} can be used. 

\section{Additional Computation for Example~\ref{EX:UNIFORM}}

 To see that $x$ is uniformly integrable in the left program and not in the 
 right one (within the inner loop), 
 imagine the inner loop as a stand-alone program and let $X_n$ be the value of 
 variable $x$ after $n$ steps of this stand-alone program (i.e., when the loop 
 terminates $x$ no longer changes). Solving a simple linear recurrence shows 
 that in the right program $\E[X_n] \rightarrow 0$ as $n\rightarrow \infty$, 
 which in particular shows uniform integrability of $X_0,X_1,X_2,\dots$. On the 
 other hand, in the left program for each $K>0$ we have $\probm(X_K \geq 
 2^K\cdot{x_0})\geq \frac{1}{2^K}$, where $x_0$ is the value of $x$ upon 
 entering the loop. Hence $\E[|X_K|cdot\indicator{X_K 
 	\geq K}]\geq 1$ for each $K$ sufficiently large, which is incompatible with 
 uniform integrability.

\section{Experimental Results}\label{sec:app_ex}
In Table~\ref{tab:exp1} in Section~\ref{sec:experiments} we presented results for twenty five benchmarks.
We now present additional result for other twenty seven benchmarks in Table~\ref{tab:exp3}.

\begin{center}
\begin{table}[]
  \centering
   \begin{tabular}{c|c|c|c|c|c}
    
{Benchmark} & {Time (s)} & {Solution} & {Dimension} & {Prob. loops} & {Prob. Assignments} \\\hline \hline
{aaron2} & {0.09} & {yes} & {2} & {yes} & {yes} \\\hline
{ax} & {0.11} & {yes} & {3} & {yes} & {yes} \\\hline
{complex} & {0.1} & {yes} & {1} & {yes} & {yes} \\\hline
{counterex1b} & {0.1} & {yes} & {3} & {yes} & {yes} \\\hline
{cousot9} & {0.09} & {no} & {} & {yes} & {no} \\\hline
{easy2} & {0.09} & {yes} & {2} & {yes} & {yes} \\\hline
{loops} & {0.1} & {no} & {} & {no} & {no} \\\hline
{\multirow{2}{*}{nestedloop}} & {0.16} & {yes} & {3} & {yes} & {no} \\\cline{2-6}
{} & {0.23} & {no} & {} & {yes} & {yes} \\\hline
{perfect1} & {0.1} & {yes} & {3} & {yes} & {yes} \\\hline
{random1d} & {0.1} & {yes} & {2} & {yes} & {yes} \\\hline
{realheapsort} & {0.22} & {no} & {} & {no} & {no} \\\hline
{realheapsort\_step1} & {0.11} & {no} & {} & {no} & {no} \\\hline
{realheapsort\_step2} & {0.16} & {no} & {} & {yes} & {no} \\\hline
{rsd} & {0.1} & {yes} & {1} & {yes} & {yes} \\\hline
{sipma91} & {0.14} & {yes} & {2} & {yes} & {yes} \\\hline
{speedFails1} & {0.09} & {yes} & {2} & {yes} & {yes} \\\hline
{speedDis1} & {0.09} & {no} & {} & {no} & {no} \\\hline
{\multirow{2}{*}{speedFails2}} & {0.14} & {yes} & {1} & {yes} & {no} \\\cline{2-6}
{} & {0.08} & {no} & {} & {yes} & {yes} \\\hline
{speedFails4} & {0.09} & {no} & {} & {no} & {no} \\\hline
{speedNestedMultipleDep} & {0.1} & {yes} & {3} & {yes} & {yes} \\\hline
{speedpldi3} & {0.1} & {yes} & {3} & {yes} & {yes} \\\hline
{speedSimpleMultiple} & {0.1} & {no} & {} & {no} & {no} \\\hline
{speedSingleSingle} & {0.11} & {yes} & {2} & {yes} & {yes} \\\hline
{terminate} & {0.09} & {yes} & {2} & {yes} & {yes} \\\hline
{wcet0} & {0.11} & {yes} & {2} & {yes} & {yes} \\\hline
{wcet2} & {0.09} & {yes} & {2} & {yes} & {yes} \\\hline
{wise} & {0.09} & {no} & {} & {no} & {no} \\\hline
    \end{tabular}
    \caption{Additonal experimental results for benchmarks from~\cite{ADFG10:lexicographic} extended with probabilistic loops and/or probabilistic assignments.}
    \label{tab:exp3}
\end{table}
\end{center}




\end{document}